\documentclass{lmcs}
\pdfoutput=1

\usepackage{lastpage}
\lmcsdoi{21}{4}{13}
\lmcsheading{}{\pageref{LastPage}}{}{}%
{Dec.~27,~2023}{Oct.~28,~2025}{}

\usepackage[utf8]{inputenc} 
\usepackage{float} 


\usepackage{graphicx}
\usepackage{xspace}
\usepackage{amssymb }
\usepackage{bbm}
\usepackage{verbatim}
\usepackage{wasysym}
\usepackage{hyperref}
\usepackage{booktabs}
\usepackage{tikz}
\usetikzlibrary{arrows.meta}
\usetikzlibrary{automata,positioning}

\usepackage{amssymb,mathtools}
\newcommand{\set}[1]{\{ #1 \}}
\begin{document}

\title[A Hierarchy of Nondeterminism]{A Hierarchy of Nondeterminism\rsuper*}
\titlecomment{\lsuper*A preliminary version of this paper appears in the Proceedings of the 46th International Symposium on Mathematical Foundations of Computer Science, 2021.}

\author[B.~Abu~Radi]{Bader Abu Radi\lmcsorcid{0000-0001-8138-9406}}
\author[O.~Kupferman]{Orna Kupferman\lmcsorcid{0000-0003-4699-6117}}
\author[O.~Leshkowitz]{Ofer Leshkowitz\lmcsorcid{0000-0001-9225-2325}}

\address{School of Engineering and Computer Science, Hebrew University, Jerusalem, Israel}

\newtheorem{remark}[thm]{Remark} 

\newcommand{\buchi}{B\"uchi\xspace}
\newcommand{\A}{{\mathcal A}}
\newcommand{\B}{{\mathcal B}}
\newcommand{\C}{{\mathcal C}}
\newcommand{\D}{{\mathcal D}}
\newcommand{\U}{{\mathcal U}}
\renewcommand{\S}{{\mathcal S}}
\newcommand{\T}{{\mathcal T}}
\newcommand{\M}{{\mathcal M}}
\newcommand{\F}{{\mathcal F}}
\newcommand{\K}{{\mathcal K}}
\newcommand{\Pcal}{{\mathcal P}}
\newcommand\floor[1]{\lfloor#1\rfloor}
\newcommand{\zug}[1]{\left\langle #1 \right\rangle}
\renewcommand{\set}[1]{\left\{ #1 \right\}}
\newcommand{\G}{{\mathcal G}}
\renewcommand{\P}{\mathbb{P}}
\renewcommand{\Pr}{\mathrm{Pr}}
\newcommand{\dist}{\mathrm{dist}}
\newcommand{\last}{\mathrm{last}}
\newcommand{\cs}{\mbox{\it co-safe}}
\renewcommand{\S}{{\mathcal S}}
\newcommand{\N}{{\mathcal N}}
\newcommand{\dxw}{DXW\xspace}
\newcommand{\dxws}{DXWs\xspace}
\newcommand{\nxw}{NXW\xspace}
\newcommand{\nxws}{NXWs\xspace}

\begin{abstract}
We study three levels in a hierarchy of nondeterminism: A nondeterministic automaton $\A$ is {\em determinizable by pruning\/} (DBP) if we can obtain a deterministic automaton equivalent to $\A$ by removing some of its transitions. Then, $\A$ is {\em history deterministic\/} (HD) if its nondeterministic choices can be resolved in a way that only depends on the past. Finally, $\A$ is {\em semantically deterministic\/} (SD) if different nondeterministic choices in $\A$ lead to equivalent states. Some applications of automata in formal methods require deterministic automata, yet in fact can use automata with some level of nondeterminism. For example, DBP automata are useful in the analysis of online algorithms, and HD automata are useful in synthesis and control. For automata on finite words, the three levels in the hierarchy coincide. We study the hierarchy for \buchi, co-\buchi, and weak automata on infinite words. We show that the hierarchy is strict, study the expressive power of the different levels in it, as well as the complexity of deciding the membership of a language in a given level. Finally, we describe a probability-based analysis of the hierarchy, which relates the level of nondeterminism with the probability that a random run on a word in the language is accepting. 
We relate the latter to nondeterministic automata that can be used when reasoning about probabilistic systems.
\end{abstract}
\maketitle
\section{Introduction}
\label{intro}
{\em Nondeterminism\/} is a fundamental notion in theoretical computer science. It allows a computing machine to examine several possible actions simultaneously. 
For automata on finite words, nondeterminism does not increase the expressive power, yet it leads to an exponential succinctness~\cite{RS59}. 

A prime application of automata theory is specification, verification, and synthesis of reactive systems~\cite{VW94,Kup18}.  Since we care about the on-going behavior of nonterminating  systems, the automata run on infinite words. Acceptance in such automata is determined according to the set of states that are visited infinitely often along the run. In {\em B\"uchi\/} automata~\cite{Buc62},
the acceptance condition is a subset $\alpha$ of states, and a run is accepting iff it visits $\alpha$ infinitely often. Dually, in {\em co-B\"uchi\/} automata, 
a run is accepting iff it visits $\alpha$ only finitely often. We also consider {\em weak\/} automata, which are a special case of both \buchi and co-\buchi automata in which no cycle contains both states in $\alpha$ and states not in $\alpha$. We use three-letter acronyms in $\{ \text{D, N} \} \times \{ \text{F, B, C, W}\} \times \{\text{W}\}$ to describe the different classes of automata. The first letter stands for the branching mode of the automaton (deterministic or nondeterministic); the second for the acceptance condition type (finite, \buchi, co-\buchi or weak); and the third indicates that we consider automata on words.

For automata on infinite words, nondeterminism may increase the expressive power and also leads to an exponential succinctness. For example, NBWs are strictly more expressive than DBWs~\cite{Lan69}, whereas NCWs are as expressive as DCWs~\cite{MH84}. In some applications of the automata-theoretic approach, such as model checking, algorithms can be based on nondeterministic automata, whereas in other applications, such as synthesis and control, they cannot. There, the advantages of nondeterminism are lost, and algorithms involve a complicated determinization construction~\cite{Saf88} or acrobatics for circumventing determinization~\cite{KV05c}.
Essentially, the inherent difficulty of using nondeterminism in synthesis and control lies in the fact that each guess of the nondeterministic automaton should accommodate all possible futures. 

A study of nondeterministic automata that can resolve their nondeterministic choices in a way that only depends on the past started in~\cite{KSV06}, where the setting is modeled by means of tree automata for derived languages. It then continued by means of  {\em history deterministic\/}  (HD) automata~\cite{HP06}.
\footnote{The notion used in \cite{HP06} is {\em good for games} (GFG) automata, as they address the difficulty of playing games on top of a nondeterministic automaton. As it turns out, the property of being good for games varies in different settings and HD is good for applications beyond games. Therefore, we use the term {\em history determinism}, introduced by Colcombet in the setting of quantitative automata with cost functions \cite{Col09}.} 
A nondeterministic automaton $\A$ over an alphabet $\Sigma$ is HD if there is a strategy $g$ that maps each finite word $u \in \Sigma^*$ to the transition to be taken after $u$ is read; and following $g$ results in accepting all the words in the language of $\A$. Note that a state $q$ of $\A$ may be reachable via different words, and $g$ may suggest different transitions from $q$ after different words are read. Still, $g$ depends only on the past, namely on the word read so far. Obviously, there exist HD automata: deterministic ones, or nondeterministic ones that are {\em determinizable by pruning\/} (DBP); that is, ones that embody an equivalent deterministic automaton. In fact, the HD automata constructed in~\cite{HP06} are DBP.
\footnote{As explained in~\cite{HP06}, the fact that the HD automata constructed there are DBP does not contradict their usefulness in practice, as their transition relation is simpler than the one of the embodied deterministic automaton and it can be defined symbolically.} 
Beyond the theoretical interest in DBP automata, they are used for modelling online algorithms: by relating the ``unbounded look ahead''
of optimal offline algorithms with nondeterminism, and relating the ``no look ahead'' of online algorithms with
determinism, it is possible to reduce questions about the competitive ratio of online algorithms and the memory
they require to questions about DBPness~\cite{AKL10,AK15}. 
 
For automata on finite words, HD-NFWs are always DBP~\cite{KSV06,Mor03}. For automata on infinite words, HD-NBWs and HD-NCWs are as expressive as DBWs and DCWs, respectively \cite{KSV06,NW98}, but they need not be DBP~\cite{BKKS13}. 
Moreover, the best known determinization construction for HD-NBWs is quadratic, and determinization of HD-NCWs has a tight exponential blow-up~\cite{KS15}. Thus, HD automata on infinite words are (possibly even exponentially) more succinct than deterministic ones.

Further research studies characterization, typeness, complementation, and further constructions and decision procedures for HD automata~\cite{KS15,BKS17,BK18}, as well as an extension of the HD setting to pushdown $\omega$-automata~\cite{LZ20} and to alternating automata~\cite{BL19,BKLS20}.

A nondeterministic automaton is {\em semantically deterministic} (SD, for short) if its nondeterministic choices lead to states with the same language. Thus, for every state $q$ of the automaton and letter $\sigma \in \Sigma$, all the $\sigma$-successors of $q$ have the same language. 
Beyond the fact that semantically determinism is a natural relaxation of determinism, and thus deserves consideration, SD automata naturally arise in the setting of HD automata. Indeed, if different $\sigma$-successors of a state in an HD automaton have different languages, then we can prune the transitions to states whose language is contained in the language of another $\sigma$-successor. 
Moreover, since containment for HD automata can be checked in polynomial time, such a pruning can be done in polynomial time \cite{HKR02,KS15}. Accordingly (see more in Section~\ref{awsn}), we can assume that all HD automata are SD~\cite{KS15}.

Hence, we obtain the following hierarchy, from deterministic to nondeterministic automata, where each level is a special case of the levels to its right. 
\begin{figure}[ht]
	\centering
	\includegraphics[width=\linewidth]{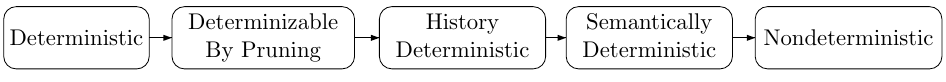}
\end{figure}

For automata on finite words, as NFWs are as expressive as DFWs, all levels of the hierarchy coincide in their expressive power. In fact, the three internal levels coincide already in the syntactic sense: every SD-NFW is DBP.  Also, given an NFW, deciding whether it is SD, HD or DBP, can each be done in polynomial time~\cite{AKL10}. 

For B\"uchi and co-B\"uchi automata, the picture is less clear, and is the subject of our research. Before we describe our results, let us mention that an orthogonal level of nondeterminism is that of {\em unambiguous\/} automata, namely automata that have a single accepting run on each word in their languages. An unambiguous NFW for a non-empty language is SD iff it is deterministic, and a DBP-NFW need not be unambiguous. It is known, however, that an HD unambiguous NCW, or NBW, is DBP~\cite{BKS17}.

We study the following aspects and questions about the hierarchy. 

\paragraph{Strictness} For each nondeterminism level we study for which acceptance condition it is a strict super class of its preceding level. Recall that not all HD-NBWs and HD-NCWs are DBP~\cite{BKKS13}, and examples for this include also SD automata. On the other hand, all HD-NWWs (in fact, all HD-\nxws whose language can be recognized by a DWW) are DBP~\cite{BKS17}. We show that SD-\nxws need not be HD for all ${\rm X} \in \{{\rm B,C,W}\}$. Of special interest is our result on weak automata, whose properties typically agree with these of automata on finite words. Here, while all SD-NFWs are HD, this is not the case for SD-NWWs. 

\paragraph{Expressive power} It is known that for all ${\rm X} \in \{{\rm B,C,W}\}$, HD-\nxws are as expressive as \dxws. We extend this result to semantic determinism and show that while SD-\nxws need not be HD, they are not more expressive, thus SD-\nxws are as expressive as \dxws. Since an SD-\nxw need not be HD, this extends the known frontier of nondeterministic B\"uchi and weak automata that are not more expressive than their deterministic counterpart. 

\paragraph{Deciding the nondeterminism level of an automaton} It is already known that deciding the HDness of a given \nxw, for ${\rm X} \in \{{\rm B,C,W}\}$,  can be done in polynomial time~\cite{AKL10, KS15, BK18}. On the other hand, deciding whether a given NCW is DBP is NP-complete~\cite{KM18}. We complete the picture in three directions. First, we show that NP-completeness of deciding DBPness applies also to NBWs. Second, we show that in both cases, hardness applies even when the given automaton is HD. Thus, while it took the community some time to get convinced that not all HD automata are DBP, in fact it is NP-complete to decide whether a given HD-NBW or HD-NCW is DBP.  Third, we study also the problem of deciding whether a given \nxw is SD, and show that it is PSPACE-complete. Note that our results imply that the nondeterminism hierarchy is not monotone with respect to complexity: deciding DBPness, which is closest to determinism, is NP-complete, then HDness can be checked in polynomial time, and finally SDness is PSPACE-complete. Also, as PSPACE-hardness of checking SDness applies already to NWWs, we get another, even more surprising, difference between weak automata and automata on finite words. Indeed, for NFWs, all the three levels of nondeterminism coincide and SDness can be checked in polynomial time.  

\paragraph{Approximated determinization by pruning}
We say that a nondeterministic automaton $\A$ is {\em almost-DBP} if $\A$ embodies a deterministic automaton $\A'$ such that the probability of a random word to be in $L(\A)\setminus L(\A')$ is $0$. Clearly, if $\A$ is DBP, then it is almost-DBP. 
Recall that DBPness is associated with the ability of an online algorithm to perform as good as an offline one. A typical analysis of the performance of an online algorithm compares its performance with that of an offline algorithm. 
The notion of almost-DBPness captures cases where the online algorithm performs, with probability $1$, as good as the offline algorithm. From the point of view of formal language theory, almost-DBPness captures the ability to approximate from below the language of automata that need not be DBP.  
We study the almost-DBPness of HD and SD automata. Our results imply that the nondeterminism hierarchy ``almost collapses":
We show that for \buchi (and hence also weak) automata, semantic determinism implies almost-DBPness, thus every SD-NBW is almost-DBP. Then, for co-\buchi automata, semantic determinism is not enough, and we need HDness. Thus, there is an SD-NCW that is not almost-DBP, yet all HD-NCWs are almost-DBP. 

\paragraph{Reasoning about Markov decision processes}
Another application in which nondeterminism is problematic is reasoning about probabilistic systems. 
Technically, such a reasoning involves the product of a Markov decision processes that models the system with an automaton for the desired behavior. When the automaton is nondeterministic, the product need not reflect the probability in which the system satisfies the behavior.  A nondeterministic automaton is {\em good-for-MDPs} (GFM) if its product  with Markov decision processes maintains the probability of acceptance. Thus, GFM automata can replace deterministic automata when reasoning about probabilistic systems \cite{HPSSTWW20,STZ22}. 
We study the relation between semantic determinism and GFMness. We show that all GFM automata are almost-DBP. On the other hand, semantic determinism implies GFMness only in B\"uchi automata (that is, all SD-NBWs are GFM), and a GFM automaton need not be SD, even for weak automata. 

\section{Preliminaries}
\label{prelim}

\subsection{Automata}
For a finite nonempty alphabet $\Sigma$, an infinite {\em word\/} $w = \sigma_1 \cdot \sigma_2 \cdots \in \Sigma^\omega$ is an infinite sequence of letters from $\Sigma$. 
A {\em language\/} $L\subseteq \Sigma^\omega$ is a set of infinite words. For $i,j\geq 0$, we use $w[1, i]$ to denote the (possibly empty) prefix $\sigma_1\cdot \sigma_2 \cdots  \sigma_i$ of $w$, use $w[i+1, j]$ to denote the (possibly empty) infix $\sigma_{i+1}\cdot \sigma_{i+2} \cdots \sigma_j$ of $w$, and use $w[i+1, \infty]$ to denote its suffix $\sigma_{i+1} \cdot  \sigma_{i+2} \cdots$. We sometimes refer also to languages of finite words, namely subsets of $\Sigma^*$. We denote the empty word by $\epsilon$.

A \emph{nondeterministic automaton} over infinite words is $\A = \langle \Sigma, Q, q_0, \delta, \alpha  \rangle$, where $\Sigma$ is an alphabet, $Q$ is a finite set of \emph{states}, $q_0\in Q$ is an \emph{initial state}, $\delta: Q\times \Sigma \to 2^Q \setminus \emptyset$ is a \emph{transition function}, and $\alpha$ is an \emph{acceptance condition}, to be defined below. 
For states $q$ and $s$ and a letter $\sigma \in \Sigma$, we say that $s$ is a $\sigma$-successor of $q$ if $s \in \delta(q,\sigma)$.

Note that $\A$ is \emph{total}, in the sense that it has at least one successor for each state and letter.
If $|\delta(q, \sigma)| = 1$ for every state $q\in Q$ and letter $\sigma \in \Sigma$, then $\A$ is \emph{deterministic}.

A \emph{run}  of $\A$ on $w = \sigma_1 \cdot \sigma_2 \cdots \in \Sigma^\omega$ is an infinite sequence of states $r = r_0,r_1,r_2,\ldots \in Q^\omega$, such that $r_0 = q_0$, and for all $i \geq 0$, we have that $r_{i+1} \in \delta(r_i, \sigma_{i+1})$. 
We extend $\delta$ to sets of states and finite words in the expected way. Thus, $\delta(S, u)$ is the set of states that $\A$ may reach when it reads the word $u \in \Sigma^*$ from some state in $S \in 2^Q$. Formally, $\delta: 2^Q\times \Sigma^* \to 2^Q$ is such that for every $S \in 2^Q$, finite word $u\in \Sigma^*$, and letter $\sigma\in \Sigma$, we have that $\delta(S, \epsilon) = S$, $\delta(S, \sigma) = \bigcup_{s\in S}\delta(s, \sigma)$, and $\delta(S, u \cdot \sigma) = \delta(\delta(S, u), \sigma)$. 
The transition function $\delta$ induces a transition relation $\Delta \subseteq Q\times \Sigma \times Q$, where for every two states $q,s\in Q$ and letter $\sigma\in \Sigma$, we have that $\langle q, \sigma, s \rangle \in \Delta$ iff $s\in \delta(q, \sigma)$. 
For a state $q\in Q$ of $\A$, we define $\A^q$ to be the automaton obtained from $\A$ by setting the initial state to be $q$. Thus, $\A^q = \langle \Sigma, Q, q, \delta, \alpha \rangle$. Two states $q,s\in Q$ are \emph{equivalent}, denoted $q \sim_{\A} s$, if $L(\A^q) = L(\A^s)$. 

The acceptance condition $\alpha$ determines which runs are ``good''. We consider here the \emph{\buchi} and \emph{co-\buchi} acceptance conditions, where $\alpha \subseteq Q$ is a subset of states. We use the terms {\em $\alpha$-states\/} and  {\em $\bar{\alpha}$-states\/} to refer to states in $\alpha$ and in $Q \setminus \alpha$, respectively. For a run $r$, let ${\it inf}(r)\subseteq Q$ be the set of states that $r$ traverses infinitely often. Thus, 
${\it inf}(r) = \{  q \in Q: q = r_i \text{ for infinitely many $i$'s}   \}$. 
A run $r$ of a \buchi automaton is \emph{accepting} iff it visits states in $\alpha$ infinitely often, thus ${\it inf}(r)\cap \alpha \neq \emptyset$. Dually, a run $r$ of a co-\buchi automaton is accepting iff it visits states in $\alpha$ only finitely often, thus ${\it inf}(r)\cap \alpha = \emptyset$.
A run that is not accepting is \emph{rejecting}.  Note that as $\A$ is nondeterministic, it may have several runs on a word $w$. The word $w$ is accepted by $\A$ if there is an accepting run of $\A$ on $w$. The language of $\A$, denoted $L(\A)$, is the set of words that $\A$ accepts. Two automata are \emph{equivalent} if their languages are equivalent.

Consider a directed graph $G = \langle V, E\rangle$. A \emph{strongly connected set\/} in $G$ (SCS, for short) is a set $C\subseteq V$ such that for every two vertices $v, v'\in C$, there is a path from $v$ to $v'$. A SCS is \emph{maximal} if it is maximal w.r.t containment, that is, for every non-empty set $C'\subseteq V\setminus C$, it holds that $C\cup C'$ is not a SCS. The \emph{maximal strongly connected sets} are also termed \emph{strongly connected components} (SCCs, for short). The \emph{SCC graph of $G$} is the graph defined over the SCCs of $G$, where there is an edge from an SCC $C$ to another SCC $C'$ iff there are two vertices $v\in C$ and $v'\in C'$ with $\langle v, v'\rangle\in E$. A SCC is \emph{ergodic} iff it has no outgoing edges in the SCC graph. 

An automaton $\A = \langle \Sigma, Q, q_0, \delta, \alpha\rangle$ induces a directed graph $G_{\A} = \langle Q, E \rangle$, where $\langle q, q'\rangle\in E$ iff there is a letter $\sigma \in \Sigma$ such that $\langle q, \sigma, q'\rangle \in \Delta$. The SCSs and SCCs of $\A$ are those of $G_{\A}$. 

The \emph{$\alpha$-free SCCs} of $\A$ are the SCCs of $\A$ that do not contain states from $\alpha$.

A B\"uchi automaton $\A$ is {\em weak}~\cite{MSS88} if for each SCC $C$ in $G_\A$, either $C\subseteq \alpha$ (in which case we say that $C$ is an accepting SCC) or $C \cap \alpha = \emptyset$ (in which case we say that $C$ is a rejecting SCC). Note that a weak automaton can be viewed as both a B\"uchi and a co-B\"uchi automaton, as a run of $\A$ visits $\alpha$ infinitely often, iff it gets trapped in an accepting SCC, iff
it visits states in $Q \setminus \alpha$ only finitely often.

We denote the different classes of word automata by three-letter acronyms in $\{ \text{D, N} \} \times \{ \text{F, B, C, W}\} \times \{\text{W}\}$. The first letter stands for the branching mode of the automaton (deterministic or nondeterministic), and  the second for the acceptance condition type (finite, \buchi, co-\buchi or weak). For example, NBWs are nondeterministic \buchi word automata.

\subsection{Probability}
Consider the probability space $(\Sigma^\omega,\P)$ where each word $w=\sigma_1\cdot \sigma_2 \cdot \sigma_3\cdots \in \Sigma^\omega$ is drawn by taking the $\sigma_i$'s to be independent and identically distributed $\mathrm{Unif}(\Sigma)$. Thus. for all positions $i \geq 1$ and letters $\sigma \in \Sigma$, the probability that  $\sigma_i$ is $\sigma$ is $\frac{1}{|\Sigma|}$. Let $\A=\zug{\Sigma,Q,q_0,\delta,\alpha}$ be a deterministic automaton, and let $G_\A=\zug{Q,E}$ be its  induced directed graph. A \emph{random walk on $\A$}, is a random walk on $G_\A$ with the probability matrix $P(q,p)=\frac{|\{\sigma\in \Sigma : \zug{q,\sigma,p}\in \Delta\} |}{|\Sigma|}$. It is not hard to see that $\P(L(\A))$ is precisely the probability that a random walk on $\A$ is an accepting run. Note that with probability $1$, a random walk on $\A$ reaches an ergodic SCC $C\subseteq Q$, where it visits all states infinitely often. It follows that $\P(L(\A))$ equals the probability that a random walk on $\A$ reaches an ergodic accepting SCC.

\subsection{Automata with Some Nondeterminism}
\label{awsn}

Consider a nondeterministic automaton $\A = \langle \Sigma, Q, q_0, \delta, \alpha  \rangle$.
The automaton $\A$ is {\em determinizable by pruning} (DBP) if we can remove some of the transitions of $\A$ and get a deterministic automaton $\A'$ that recognizes $L(\A)$.  We then say that $\A'$ is a {\em deterministic pruning\/} of $\A$. 

The automaton $\A$ is \emph{history deterministic} (\emph{HD}, for short) if its nondeterminism can be resolved based on the past, thus on the prefix of the input word read so far. Formally, $\A$ is \emph{HD} if there exists a {\em strategy\/} $f:\Sigma^* \to Q$ such that the following hold: 
\begin{enumerate}
	\item 
	The strategy $f$ is consistent with the transition function. That is, $f(\epsilon)=q_0$, and for every finite word $u \in \Sigma^*$ and letter $\sigma \in \Sigma$, we have that $\zug{f(u),\sigma,f(u \cdot \sigma)} \in \Delta$. 
	\item
	Following $f$ causes $\A$ to accept all the words in its language. That is, for every infinite word $w = \sigma_1 \cdot \sigma_2 \cdots \in \Sigma^\omega$, if $w \in L(\A)$, then the run $f(w[1, 0]), f(w[1, 1]), f(w[1, 2]), \ldots$, which we denote by $f(w)$, is 
	an accepting run of $\A$ on $w$. 
\end{enumerate}
We say that the strategy $f$ \emph{witnesses} $\A$'s HDness. 
Note that a strategy $f$ need not use all the states or transitions of $\A$. 
In particular, a state $q$ of $\A$ may be such that $\A^q$ is not HD. We say that a state $q$ of $\A$ is \emph{HD} if $\A^q$ is HD.

The automaton $\A$ is \emph{semantically deterministic\/} (SD, for short) if different nondeterministic choices in $\A$ lead to equivalent states. Thus, for every state $q\in Q$ and letter $\sigma \in \Sigma$, all the $\sigma$-successors of $q$ are equivalent: for every two states $s, s'\in \delta(q,\sigma)$, we have that $s \sim_{\A} s'$.  

Note that every deterministic automaton $\A$ is DBP (with $\A$ being the deterministic pruning of itself).  Also, every nondeterministic DBP automaton is HD. Indeed, the deterministic pruning of $\A$ induces a witness strategy. Finally, if a state $q$ of an HD automaton has a $\sigma$-successor $s$ that does not accept all the words $w$ such that $\sigma \cdot w$ is in $L(\A^q)$, then the $\sigma$-transition from $q$ to $s$ can be removed. Moreover, by ~\cite{HKR02,KS15, BK18}, the detection of such transitions can be done in polynomial time. Hence, we can assume that every HD automaton, and hence also every DBP automaton, is SD.\footnote{Using the terminology of \cite{KS15}, we can assume that the {\em residual languages\/} of all the $\sigma$-successors of each state in an HD automaton are equivalent.}

For ${\rm X} \in \{{\rm B,C,W}\}$, we use SD-NXW, HD-NXW, and DBP-NXW to refer to the different classes of NXWs. For example, SD-NBWs are nondeterministic B\"uchi automata that are semantically-deterministic.

\section{The Syntactic and Semantic Hierarchies}

In this section we study syntactic and semantic hierarchies induced by the different levels of nondeterminism. We start with the syntactic hierarchy. For two classes $\C_1$ and $\C_2$ of automata, we use $\C_1 \preceq \C_2$ to indicate that every automaton in $\C_1$ is also in $\C_2$. Accordingly, $\C_1 \prec \C_2$ if $\C_1 \preceq \C_2$ yet there are automata in $\C_2$ that are not in $\C_1$. For example, clearly, DFW $\prec$ NFW. We first show that the nondeterminism hierarchy is strict, except for all HD-NWWs being DBP. The latter is not surprising, as all HD-NFWs are DBP, and weak automata share many properties with automata on finite words. On the other hand, unlike the case of finite words, we show that not all SD-NWWs are HD. In fact the result holds already for NWWs that accept co-safety languages, namely all whose states except for an accepting sink are rejecting. 

\begin{thm}{\bf [Syntactic Hierarchy]}
	For ${\rm X} \in \{{\rm B,C,W}\}$, we have that {\rm DXW} $\prec$ {\rm DBP-NXW} $\preceq$ {\rm HD-NXW} $\prec$ {\rm SD-NXW} $\prec$ {\rm NXW}. For ${\rm X} \in \{{\rm B,C}\}$, the second inequality is strict.
\end{thm}

\begin{proof}
	By definition, each class is a special case of the one to its right. We prove strictness. 
	It is easy to see that the first and last strict inequalities hold. For the first strict inequality, for all ${\rm X} \in \{{\rm B,C,W}\}$, consider a nonempty \dxw $\A$, and obtain an \nxw $\B$ from $\A$ by duplicating some state and the transitions to it. Then, $\B$ is a DBP-\nxw that is not a \dxw. For the last inequality, consider an SD-\nxw $\A$, and let $\sigma$ be a letter such that $\A$ accepts some word that starts with $\sigma$. Consider the \nxw $\C$ obtained from $\A$ by adding a $\sigma$-transition from $\A$'s initial state to a new rejecting sink. Then, as at least one $\sigma$-successor of the initial state of $\A$ is not empty, we have that $\C$ is an \nxw that is not an SD-\nxw.
	
	The relation between DBPness and HDness has already been studied. It is shown in~\cite{BKKS13} that HD-\nxw need not be DBP for ${\rm X} \in \{{\rm B,C}\}$, and shown in \cite{BKS17}
	that HD-NWW are DBP.

	It is left to relate HDness and SDness. Consider the NWW ${\mathcal W}$ in Figure~\ref{fig weak}. It is not hard to check that ${\mathcal W}$ is weak, and note that it is SD, as all its states recognize the language $\{a,b\}^\omega$.  Yet $\mathcal{W}$ is not HD, as every strategy has a word with which it does not reach $q_{acc}$ -- a word that forces each visit in $q_a$ and $q_b$ to be followed by a visit in $q_0$.    
	
	\begin{figure}[ht]
		\centering
		\includegraphics[width=0.5\textwidth]{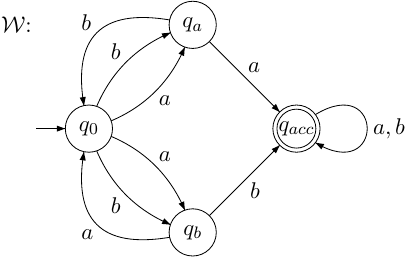}
		\caption{An SD-NWW that is not HD.} 
		\label{fig weak}
	\end{figure}
	
	Hence, HD-NWW $\prec$ SD-NWW. As weak automata are a special case of \buchi and co-\buchi, strictness for them follows. 
\end{proof}

We continue to the semantic hierarchy, where we study the expressive power of the different classes. Now,  for two classes $\C_1$ and $\C_2$ of automata, we say that $\C_1$ is less expressive than $\C_2$, denoted $\C_1 \leq \C_2$, if every automaton in $\C_1$ has an equivalent automaton in $\C_2$. Accordingly, $\C_1=\C_2$ if $\C_1\leq \C_2$ and $\C_2\leq \C_1$, and $\C_1<\C_2$ if $\C_1\leq \C_2$ yet $\C_2\neq\C_1$.
Since NCW=DCW, we expect the hierarchy to be strict only in the cases of \buchi and weak automata. As we now show, however, semantically deterministic automata are not more expressive than deterministic ones also in the case of B\"uchi and weak automata. 

\begin{thm}{\bf [Semantic Hierarchy]} 
\label{sem hie}
	For $X \in \{B,W\}$, we have that {\rm DXW} $=$ {\rm DBP-NXW} $=$ {\rm HD-NXW} $=$ {\rm SD-NXW} $<$ {\rm NXW}.
\end{thm}

\begin{proof}
	In~\cite{KSV06,KS15}, the authors suggest variants of the subset construction that determinize HD-NBWs. 
	As we argue below, the construction in \cite{KS15} is correct also when applied to SD-NBWs. Moreover, it preserves weakness.  Thus, DBW$=$SD-NBW and DWW$=$SD-NWW. Also, the last inequality follows from the fact DBW$<$NBW and DWW$<$NWW \cite{Lan69}. 
	
	Given an NBW $\A = \zug{\Sigma, Q, q_0, \delta, \alpha}$, the DBW $\A' = \zug{\Sigma, Q', q'_0, \delta', \alpha'}$ generated in \cite{KS15}\footnote{The construction in \cite{KS15} assumes automata with transition-based acceptance, and (regardless of this) is slightly different: when $\alpha$ is visited, $\A'$ continues with a single state from the set of successors. However, the key point remains the same: $\A$ being SD enables $\A'$ to maintain only subsets of states, rather than Safra trees, which makes determinization much easier.} is such that $Q'=2^Q$, $q'_0= \{q_0\}$, $\alpha' = \{S\in 2^Q: S\subseteq \alpha\}$, and the transition function $\delta'$ is defined for every subset $S\in 2^Q$ and letter $\sigma\in \Sigma$ as follows. If $\delta(S, \sigma)\cap \alpha = \emptyset$, then $\delta'(S, \sigma)=\delta(S, \sigma)$. Otherwise, namely if $\delta(S, \sigma)\cap \alpha \neq \emptyset$, then $\delta'(S, \sigma)=\delta(S, \sigma) \cap \alpha$. 

Thus, we proceed as the standard subset construction, except that whenever a constructed set contains a state in $\alpha$, we leave in the set only states in $\alpha$. Accordingly, every reachable state $S\in Q'$ contains only $\alpha$-states of $\A$ or only $\bar{\alpha}$-states of $\A$. Note that as $\A$ is SD, then for every two states  $q,q' \in Q$, letter $\sigma\in \Sigma$, and transitions $\langle q, \sigma, s\rangle,\langle q', \sigma, s'\rangle \in \Delta$, if $q \sim_{\A} q'$, then $s \sim_{\A} s'$. Consequently, every reachable state $S$ of $\A'$ consists of $\sim_\A$-equivalent states. As we formally prove in Appendix~\ref{app css}, these properties guarantee that indeed $L(\A')=L(\A)$ and that weakness of $\A$ is maintained in $\A'$. 
\end{proof}

\section{Deciding the Nondeterminism Level of an Automaton}

In this section we study the complexity of the problem of deciding the nondeterminism level of a given automaton. Note that we refer here to the syntactic class (e.g., deciding whether a given NBW is HD) and not to the semantic one (e.g., deciding whether a given NBW has an equivalent HD-NBW). Indeed, by Theorem~\ref{sem hie}, the latter boils down to deciding whether the language of a given \nxw, for ${\rm X} \in \{{\rm B,C,W}\}$, can be recognized by a \dxw, which is well known: the answer is always ``yes'' for an NCW, and the problem is PSPACE-complete for NBWs  and NWWs \cite{KV05b}.\footnote{The proof in \cite{KV05b} is for NBWs, yet the arguments there apply also for weak automata.}

Our results are summarized in Table~\ref{level}. The entries there describe the complexity of deciding whether a given \nxw belongs to a certain nondeterministic level (for example, the upper left cell describes the complexity of deciding whether an NBW is DBP). In fact, the results described are valid already when the given \nxw is known to be one level above the questioned one with only two exceptions -- deciding the DBPness of an HD-NCW, whose NP-hardness is proved in Theorem~\ref{NPC coBuchi DBP}, and deciding whether a given NWW is DBP, which is PTIME in general, and is $O(1)$ when the given NWW is HD, in which case the answer is always ``yes''. 
\begin{table}[ht!] 
	\begin{center}
		\label{tab:table1}
		\begin{tabular}{lccc}
							&\textbf{DBP}				& \textbf{HD} 		& \textbf{SD}			\\
			\midrule 
			\textbf{NBW} 	& NP-complete 				& PTIME 			& PSPACE-complete	 	\\
							& \autoref{NPC Buchi DBP}	& \cite{BK18} 		& \autoref{sd pspace}	\\
			\midrule
			\textbf{NCW} 	& NP-complete 				& PTIME 			& PSPACE-complete 		\\
							& \cite{KM18}				& \cite{KS15}	 	& \autoref{sd pspace}	\\
			\midrule
			\textbf{NWW}	& PTIME 					& PTIME 			& PSPACE-complete 		\\
							& \cite{KS15,BK18}			& \cite{KS15,BK18} 	& \autoref{sd pspace}	\\
		\end{tabular}
	\end{center}
	\caption{The complexity of deciding the nondeterminism level of an \nxw, for ${\rm X} \in \{{\rm B,C,W}\}$. 
		The results are valid also in the case the given \nxw is one level to the right in the nondeterminism hierarchy, with the exception of deciding the DBPness of an HD-NCW, whose NP-hardness is proved in Theorem~\ref{NPC coBuchi DBP}, and deciding the DBPness of an HD-NWW, where the answer is always positive \cite{BKS17}.}
	\label{level}
\end{table}

We start with the complexity of deciding semantic determinism. 

\begin{thm}
\label{sd pspace}
Deciding whether an \nxw is semantically deterministic is PSPACE-complete, for ${\rm X} \in \{{\rm B,C,W}\}$.
\end{thm}

\begin{proof}
Membership in PSPACE is easy, as we check SDness by polynomially many checks of language equivalence. Formally, given an \nxw $\A=\zug{\Sigma,Q,q_0,\delta,\alpha}$, a PSPACE algorithm goes over all states $q \in Q$, letters $\sigma$, and $\sigma$-successors $s$ and $s'$ of $q$, and checks that $s \sim_\A s'$.  Since language equivalence can be checked in PSPACE \cite{SVW87} and there are polynomially many checks to perform, we are done. 

Proving PSPACE-hardness, we do a reduction from polynomial-space Turing machines. Given a Turing machine $T$ with
$|T|$ many transitions and space complexity $s:\mathbb{N}\rightarrow\mathbb{N}$, we construct in time polynomial in $|T|$ and $s(0)$, an NWW $\A$ of size linear in $T$ and $s(0)$, such that $\A$ is SD iff $T$ accepts the empty tape\footnote{This is sufficient, as one can define a generic reduction from every language $L$ in PSPACE as follows. Let $T_L$ be a Turing machine that decides $L$ in polynomial space $f(n)$. On input $w$ for the reduction, the reduction considers the machine $T_w$ that on every input, first erases the tape, writes $w$ on its tape, and then runs as $T_L$ on $w$. Then, the reduction outputs an automaton $\A$, such that $T_w$ accepts the empty tape iff $\A$ is SD. Note that the space complexity of $T_w$ is $s(n) = \max(n ,f(|w|))$, and that $w$ is in $L$ iff $T_w$ accepts the empty tape. Since $\A$ is constructed in time polynomial in $s(0)=f(|w|)$ and $|T_w|=\mathrm{poly}(|w|)$, it follows that the reduction is polynomial in $|w|$.}. Clearly, this implies a lower bound also for NBWs and NCWs. 
Let $n_0=s(0)$. Thus, each configuration in the computation of $T$ on the empty tape uses at most $n_0$ cells. 

We assume that $T$ halts from all configurations (that is, not just from these reachable from an initial configuration of $T$); indeed, by adding to $T$ a step-counter that uses polynomial-space, one can transform a polynomial-space Turing machine that need not halt from all configurations to one that does halt. 

We also assume, without loss of generality, that once $T$ reaches a final (accepting or rejecting) state, it erases the tape, moves with its reading head to the leftmost cell, and moves to the initial state. Thus, all computations of $T$ are infinite and after visiting a final configuration for the first time, they consists of repeating the same finite computation on the empty tape
that uses $n_0$ tape cells. 

We define $\A$ so that it accepts a word $w$ iff 
\begin{itemize}
\item
(C1) $w$ is not a suffix of an encoding of a legal computation of $T$ that uses at most $n_0$ cells,
or 
\item
(C2) $w$ includes an encoding of an accepting configuration of $T$ that uses at most $n_0$ cells. 
\end{itemize}
It is not hard to see that if $T$ accepts the empty tape, then $\A$ is universal (that is, accepts all words). Indeed, each word $w$ is either not a suffix of an encoding of a legal computation of $T$ 
that uses at most $n_0$ cells, 
in which case $w$ is accepted thanks to C1, or $w$ is such a suffix, in which case, since we assume that $T$ halts from any configuration, the encoded computation eventually reaches a final configuration and, by the definition of $T$, continues by an encoding of the computation of $T$ on the empty tape. Thus, $w$ eventually includes the encoding of the computation of $T$ on the empty tape. In particular, since $T$ accepts the empty tape, $w$ includes an encoding of an accepting configuration that uses at most $n_0$ tape cells, and so $w$ is accepted thanks to C2.

Also, if $T$ rejects the empty tape, then $\A$ rejects the word $w_\varepsilon$ that encodes the computation of $T$ on the empty tape. Indeed, C1 is not satisfied, and since by definition of $T$, the infinite computation of $T$ on the empty tape is just the same finite rejecting computation of $T$ on the empty tape repeated, then $w_\varepsilon$ clearly does not include an encoding of an accepting configuration, and so C2 is not satisfied too. 

In order to define $\A$ so that it is SD iff $T$ accepts the empty tape, we define all its states to be universal iff $T$ accepts the empty tape. Intuitively, we do it by letting $\A$ guess and check the existence of an infix that witnesses  satisfaction of C1 or C2, and also let it, at each point of its operation, go back to the initial state, where it can guess again. Recall that $\A$ is universal when $T$ accepts the empty tape. That is, for every infinite word $w$, all the suffixes of $w$ satisfy C1 or C2. Thus, $\A$ making a bad guess on $w$ does not prevent it from later branching into an accepting run.

We now describe the operation of $\A$ in more detail (see \autoref{SD-Hardness}). 
\begin{figure}[ht]
	\centering
	\includegraphics[width=0.9\textwidth]{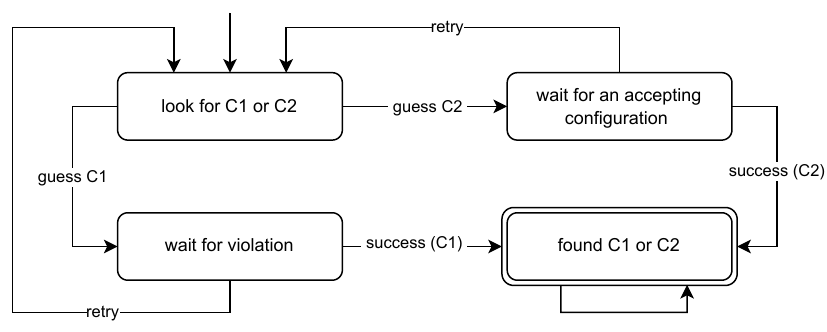}
	\caption{The structure of the NWW constructed in \autoref{sd pspace}.} 
	\label{SD-Hardness}
\end{figure}
In its initial state, $\A$ guesses 
which of C1 and C2 is satisfied. In case $\A$ guesses that C1 is satisfied, it guesses the place in which $w$ includes a violation of the encoding. As we detail in \autoref{app tm}, this amounts to guessing whether the violation is in the encoding of a single configuration encoded in $w$, or whether the violation is of the transition function of $T$. 
For the second type of violations, $\A$ may guess, in each step, that the next three letters encode a position in a configuration and the letter to come $n_0$ letters later, namely at the same position in the successive configuration, is different from the one that should appear in a legal encoding of two successive configurations. If a violation is detected, $\A$ moves to an accepting sink. Otherwise, $\A$ returns to the initial state and $w$ gets another chance to be accepted. 

In case $\A$ guesses that C2 is satisfied, it essentially waits to see a letter that encodes the accepting state of $T$, and when such a letter arrives it moves to the accepting sink.

Also, whenever $\A$ waits to witness some behavior, namely, waits for a position where there is a violation of the encoding, or waits for an encoding of an accepting configuration of $T$ that uses $n_0$ cells, it may nondeterministically, upon reading the next letter, return to the initial state.  
It is not hard to see that $\A$ can be defined in size linear in 
$|T|$ and $n_0$.
As the only accepting states of $\A$ is the accepting sink, it is clearly weak, and in fact describes a co-safety language. 

We prove that $T$ accepts the empty tape iff $\A$ is SD. First, if $T$ rejects the empty tape, then $\A$ is not SD. To see this, consider the word $w_\varepsilon$ that encodes the computation of $T$ on the empty tape, and let $w_{\varepsilon}'$ be a word that is obtained from $w_\varepsilon$ by making a single violation in the first letter. 
Note that $w_\varepsilon'\in L(\A)$ since it has a violation. Note also that any proper suffix of $w_\varepsilon'$ is also a suffix of $w_\varepsilon$, and hence is not in $L(\A)$ as it encodes a suffix of a computation of $T$ that uses at most $n_0$ tape cells and does not include an encoding of an accepting configuration of $T$. 
Consequently, the word $w_\varepsilon'$ can be accepted by $\A$ only by guessing a violation that is caused by the first letter. In particular, if we guess incorrectly that there is a violation on the second letter of $w'_\varepsilon$, then that guess fails, and after it we cannot branch to an accepting run. This shows that $\A$ is not SD.

For the other direction, we show that if $T$ accepts the empty tape, then all the states of $\A$ are universal. In particular, all the states of $\A$ are equivalent, which clearly implies that $\A$ is SD. First, as argued above, if $T$ accepts the empty tape, then $\A$ is universal. In addition, by the definition of $\A$, for every infinite word $w$ and for all states $q$ of $\A$ that are not the accepting sink, there is a path from $q$ to the initial state that is labeled by a prefix of $w$. Thus, the language of all states is universal, and we are done.

Thus, we conclude that $T$ accepts the empty tape iff $\A$ is SD.
In \autoref{app tm}, we give the full technical details of the construction of $\A$.
\end{proof}

Note that the considerations in the proof of Theorem~\ref{sd pspace} can be used in order to prove that the universality problem for SD-NBWs is PSPACE-hard \cite{AK23}.

 \begin{thm}\label{NPC Buchi DBP}
	The problem of deciding whether a given NBW or HD-NBW is DBP is NP-complete.
\end{thm}

\begin{proof}
For membership in NP, observe we can check that a witness deterministic pruning $\A'$ is equivalent to $\A$ by checking whether $L(\A)\subseteq L(\A')$. Since $\A'$ is deterministic, the latter can be checked in polynomial time. For NP-hardness, we describe a parsimonious polynomial time reduction from SAT. That is, given a CNF formula $\varphi$, we construct an HD-NBW $\A_\varphi$ such that there is a bijection between assignments to the variables of $\varphi$ and DBWs embodied in $\A_\varphi$, and an assignment satisfies $\varphi$ iff its corresponding embodied DBW is equivalent to $\A_\varphi$. In particular, $\varphi$ is satisfiable iff $\A_\varphi$ is DBP. 

Consider a SAT instance $\varphi$ over the variable set $X=\{x_1,\ldots,x_n\}$ and with $m\geq 1$ clauses $C=\{c_1,\ldots,c_m\}$. For $n \geq 1$, let $[n]=\{1,2,\ldots,n\}$. 
	For a variable $x_k \in X$, let $C^0_k \subseteq C$ be the set of clauses in which $x_k$ appears negatively, and let $C^1_k \subseteq C$ be the set of clauses in which $x_k$ appears positively. For example, if $c_1=x_1\lor \lnot x_2 \lor x_3$, then $c_1$ is in $C^1_1$, $C_2^0$, and $C^1_3$.
		Assume that all clauses depend on at least two different variables (that is, no clause is a tautology or forces an assignment to a single variable). Let $\Sigma_{n,m} = X\cup C$, and let $$R_{n,m} = (X\cdot C)^*\cdot \{x_1\cdot c_j\cdot x_2\cdot c_j\cdots x_n\cdot c_j : j\in [m]\} \subseteq \Sigma^*_{n,m}.$$
	We construct an HD-NBW $\A_\varphi$ that recognizes $L_{n,m} = (R_{n,m})^\omega$, and is DBP iff $\varphi$ is satisfiable.

	Let $D_{n,m}$ be a DFW that recognizes $R_{n,m}$ with $O(n \cdot m)$ states, a single accepting state $p$, and an initial state $q_0$ that is visited only once in all runs.  For example, we can define $D_{n,m}=\langle \Sigma_{n,m}, Q_{n,m}, q_0, \delta_{n,m}, \{p\}\rangle$ as follows: from $q_0$, the DFW expects to read only words in $(X \cdot C)^*$ -- upon a violation of this pattern, it goes to a rejecting sink. Now, if the pattern is respected, then with $X \setminus \{x_1\}$, the DFW goes to two states where it loops with $C \cdot (X \setminus \{x_1\})$ and, upon reading $x_1$ from all states that expect to see letters in $X$, it branches with each $c_j$, for all $j \in [m]$, to a path where it hopes to detect an $x_2\cdot c_j\cdots x_n\cdot c_j$ suffix. 
	If the detection is completed successfully, it goes to the accepting state $p$. Otherwise, it returns to the two-state loop.

\noindent 
Now, we define $\A_\varphi=\zug{\Sigma_{n,m}, Q_\varphi, p, \delta_\varphi, \{q_0\}}$, where $Q_\varphi=Q_{n,m} \cup \{q^i_k : (i,k)\in\{0,1\}\times[n]\}$. The idea behind $\A_\varphi$ is as follows. From state $p$ (that is, the accepting state of $D_{n,m}$, which is now the initial state of $\A_\varphi$), the NBW $\A_\varphi$ expects to read a letter in $X$. When it reads $x_k$, for $1 \leq k \leq n$, it nondeterministically branches to the states $q_k^0$ and $q_k^1$. Intuitively, when it branches to $q_k^0$, it guesses that the clause that comes next is one that is satisfied when $x_k=0$, namely a clause in $C^0_k$. Likewise,  when it branches to $q_k^1$, it guesses that the clause that comes next is one that is satisfied when $x_k=1$, namely a clause in $C^1_k$. When the guess is successful, $\A_\varphi$ moves to the $\alpha$-state $q_0$. When the guess is not successful, it returns to $p$. Implementing the above intuition, transitions from the states $Q_{n,m}\setminus \{p\}$ are inherited from $D_{n,m}$, and transitions from the states in $\{q^i_k : (i,k)\in\set{0,1}\times[n]\} \cup \{p\}$ are defined as follows (see also \autoref{sat fig}).
	\begin{itemize}
		\item For all $k\in[n]$, we have that $\delta_\varphi(p,x_k)=\{q^0_k,q^1_k\}$. 
				\item For all $k\in[n]$, $i\in {\{0,1\}}$, and $j\in[m]$, if $c_j  \in C_k^i$, then  
				$\delta_\varphi(q^i_k,c_j)=\{q_0\}$. Otherwise, $\delta_\varphi(q^i_k,c_j)=\{p\}$. For example, if $c_1=x_1\lor \lnot x_2 \lor x_3$, then $\delta_\varphi(q^0_2,c_1)=\{q_0\}$ and $\delta_\varphi(q^1_2,c_1)=\{p\}$.
					\end{itemize}
				
\begin{figure}[ht]
		\centering
		\centering
		\includegraphics[width=0.53\textwidth]{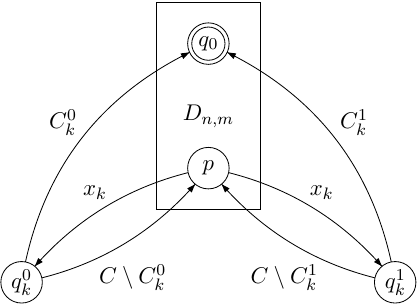}
		\caption{The transitions to and from the states $q_k^0$ and $q_k^1$ in $\A_\varphi$.}
		\label{sat fig}
	\end{figure}		
					
	\noindent 
	Note that $p$ is the only nondeterministic state of $\A_\varphi$ and that for every deterministic pruning of $\A_\varphi $, all the words in $(X\cdot C)^\omega $ have an infinite run in the pruned automaton. This run, however, may eventually loop in $\{p\}\cup \{q^0_k,q^1_k : k\in [n]\}$. 
	Note also, that for readability purposes, the automaton $\A_\varphi$ is not total. Specifically, the states of $\A_\varphi$ are partitioned into states that expect to see letters in $X$ and states that expect to see letters in $C$. In particular, all infinite paths in $\A_\varphi$ are labeled by words in $(X\cdot C)^\omega$.
Thus, when defining an HD strategy $g$ for $\A_\varphi$, we only need to define $g$ on prefixes in $(X\cdot C)^*\cup (X\cdot C)^*\cdot X$.

In the following propositions, we prove that $\A_\varphi$ is an HD NBW recognizing $L_{n,m}$, and that $\A_\varphi$ is DBP iff $\varphi$ is satisfiable.
	\end{proof}

\begin{prop}
	$L(\A_\varphi)\subseteq L_{n,m}$.
\end{prop}
\begin{proof}
	As already mentioned, all infinite paths of $\A_\varphi$, accepting or rejecting, are labeled by words in $(X\cdot C)^\omega$. 
	Further, any accepting run of $\A_\varphi$ has infinitely many sub-runs that are accepting finite runs of $D_{n,m}$. 
		Let $\infty R_{n,m}$ be the language of infinite words that include infinitely many disjoint finite subwords in $R_{n,m}$. 
	Since $L_{n,m}=(R_{n,m})^\omega = (X\cdot C)^\omega\cap (\infty R_{n,m})$, it follows that $L(\A_\varphi)\subseteq L_{n,m}$. 
\end{proof}

\begin{prop}

\label{A phi HD}
	There exists a strategy $g:\Sigma^*\rightarrow Q_{\varphi}$ for $\A_\varphi$ that accepts all words in $L_{n,m}$. Formally, for all $w\in L_{n,m}$, the run $g(w)=g(w[1,0]),g(w[1,1]),g(w[1,2]),\ldots$, is an accepting run of $\A_\varphi$ on $w$.
\end{prop}

\begin{proof}
	The definition of $L_{n,m}$ is such that when reading a prefix that ends with a subword of the form $x_1\cdot c_j$, for some $j\in [m]$,  then we can guess that the word continues with $x_2\cdot c_j \cdot x_3 \cdot c_j \cdots x_n\cdot c_j$; thus that $c_j$ is the clause that is going to repeat. Therefore, when we are at state $p$ after reading a word that ended with $x_1 \cdot c_j$, and we read $x_2$, it is a good HD strategy to move to a state $q^i_2$ such that the assignment $x_2=i$ satisfies $c_j$ (if such $i\in \{0,1\}$ exists; otherwise the strategy can choose arbitrary between $q^0_2$ and $q^1_2$), and if the run gets back to $p$, the strategy continues with assignments that hope to satisfy $c_j$, until the run gets to $q_0$ or another occurrence of $x_1$ is detected. Note that while it is not guaranteed that for all $k\in [n]$ there is $i\in \{0,1\}$ such that the assignment $x_k=i$ satisfies $c_j$, it is guaranteed that such an $i$ exists for at least two different $k$'s (we assume that all clauses depend on at least two variables). Thus, even though we a priori miss an opportunity to satisfy $c_j$  with an assignment to $x_1$, it is guaranteed that there is another $2\leq k\leq n$ such that $c_j$ can be satisfied by $x_k$. 
	 
	We define $g$ inductively as follows. Recall that $\A_\varphi$ is nondeterministic only in the state $p$, and so in all other states, the strategy $g$ follows the only possible transition. First, for all $k\in [n]$, we define $g(x_k)=q^0_k$. Let $v\in (X\cdot C)^*\cdot X$, be such that $g$ has already been defined on $v$ and let $j\in [m]$. Since $v\notin (X\cdot C)^*$, we have that $g(v)\neq p$ and so $g(v\cdot c_j)$ is uniquely defined. We continue and define $g$ on $u=v\cdot c_j\cdot x_k$, for all $k\in [n]$. If $g(v\cdot c_j)\neq p$, then $g(u)$ is uniquely defined. Otherwise, $g(v\cdot c_j)= p$ and we define $g(u)$ as follows,
	\begin{itemize}
		\item If $k=1$, then we define $g(u)=q^0_1$.
		\item If $k>1$ and $x_k$ participates in $c_j$, then we define $g(u)=q^i_k$, where $i\in \{0,1\}$ is such that $c_j\in C^i_k$. Note that since we assume that $c_j$ is not a tautology, then $i$ is uniquely defined.
		\item If $k>1$ and $x_k$ does not participate in $c_j$, then the value of $c_j$ is not affected by the assignment to $x_k$, and in that case we define $g(u)=q^0_k$.
	\end{itemize}
	\noindent 
	The reason for the distinction between the cases $k=1$ and $k>1$ is that when we see a finite word that ended with $c_j\cdot x_1$, then there is no special reason to hope that the next letter is going to be $c_j$. This is in contrast, for example, to the case we have seen a word that ends with $c_{j'}\cdot x_1\cdot c_j\cdot x_2$, where it is worthwhile to guess we are about to see $c_j$ as the next letter. 
	
	By the definition of $g$, it is consistent with $\Delta_\varphi$. In Appendix~\ref{app np HDnbw} we formally prove that $g$ is a winning HD strategy for $\A_\varphi$. Namely, that for all $w\in L_{n,m}$, the run $g(w)$ on $w$, generated by $g$ is accepting.
\end{proof}

We now examine the relation between prunings of $\A_\varphi$ and assignments to $\varphi$, and conclude the proof by proving the following:

\begin{prop}\label{dbp iff sat}
		The formula $\varphi$ is satisfiable iff the HD-NBW $\A_\varphi$ is DBP.
\end{prop}

\begin{proof}
Consider an assignment $i_1,\ldots,i_n\in \{0,1\}$, for $X$. I.e., $x_k=i_k$ for all $k\in[n]$. A possible memoryless HD strategy is to always move from $p$ to $q^{i_k}_k$ when reading $x_k$. This describes a one to one correspondence between assignments for $X$ and prunings of $\A_\varphi$. Assume that the assignment $i_k\in \{0,1\}$, for $k\in [n]$, satisfies $\varphi$. Then, the corresponding pruning recognizes $L_{n,m}$. Indeed, instead of trying to satisfy the last read clause $c_j$, we may ignore this extra information, and rely on the fact that one of the assignments $x_k=i_k$ is going to satisfy $c_j$. In other words, the satisfiability of $\varphi$ allows us to ignore the history and still accept all words in $L_{n,m}$, which makes $\A_\varphi$ DBP. On the other hand, if an assignment does not satisfy some clause $c_j$, then the corresponding pruning fails to accept the word $(x_1\cdot c_j\cdots x_n\cdot c_j)^\omega$, which shows that if $\varphi$ is not satisfiable then $\A_\varphi$ is not DBP. In Appendix~\ref{app np HDnbw} we formally prove that there is a one to one correspondence between prunings of $\A_\varphi$ and assignments to $\varphi$, and that an assignment satisfies $\varphi$ iff the corresponding pruning recognizes $L_{n,m}$, implying Proposition~\ref{dbp iff sat}.
\end{proof}

We continue to co-\buchi automata. In \cite{KM18}, the authors prove that deciding the DBPness of a given NCW is NP-complete. Below we extend their proof to HD-NCWs.

\begin{thm}\label{NPC coBuchi DBP}
	The problem of deciding whether a given HD-NCW is DBP is NP-complete.
\end{thm}

\begin{proof}
The upper bound follows by Lemma 18 in \cite{KM18} where the authors prove that checking if an automaton is DBP is in NP already for general NCWs.
For the lower bound, we analyze the reduction in \cite{KM18} and prove that the NCW constructed there is HD.
The reduction is from the \emph{Hamiltonian-cycle} problem: 
given a connected graph $G = \zug{[n], E}$, the reduction outputs an NCW $\A_G$ over the alphabet $[n]$
that is obtained from $G$ by adding self loops to all vertices, labelling the loop at a vertex $i$ by the letter $i$, and labelling the edges from vertex $i$ to all its neighbors in $G$ by every letter $j\neq i$. Then, the co-\buchi condition requires a run to eventually get stuck at a self-loop\footnote{The exact reduction is more complicated and involves an additional letter $\#$ that forces each deterministic pruning of $\A_G$ to proceed to the same neighbor of $i$ upon reading a letter $j\neq i$ from the vertex $i$.}. Accordingly, $L(\A_G) = [n]^*\cdot \bigcup_{i\in [n]} i^\omega$.

It is not hard to see that $\A_G$ is HD. Indeed, an HD strategy can decide to which neighbor of $i$ to proceed with a letter $j \neq i$ by following a cycle $c$ that traverses all the vertices of the graph $G$. Since when we read $j \neq i$ at vertex $i$ we move to a neighbor state, then by following the cycle $c$ upon reading $i^\omega$, we eventually reach the vertex $i$ and get stuck at the $i$-labeled loop. Thus, the Hamiltonian-cycle problem is reduced to DBPness of an HD-NCW, and we are done.
\end{proof}

\section{Approximated Determinization by Pruning}
\label{sec prob}

Consider a nondeterministic automaton $\A$. 
We say that $\A$ is {\em almost-DBP\/} if it can be pruned to a deterministic automaton that retains almost all the words in $L(\A)$.\footnote{Not to be confused with the definition of almost-DBPness in \cite{AKL10}, where it is used to describe automata that are deterministic in their transition function yet may have a set of initial states.}
Formally, an NXW $\A$ is almost-DBP if there is a DXW $\A'$ embodied in $\A$ such that $\P(L(\A)\setminus L(\A'))=0$. 
Thus, while $\A'$ need not accept all the words accepted by $\A$, it rejects only a negligible set of words in $L(\A)$. 
Note that if $\A$ is DBP or $\P(L(\A))=0$, then $\A$ is almost-DBP. Indeed, if $\A$ is DBP, then the DXW $\A'$ embodies in $\A$ with $L(\A')=L(\A)$ is such that $L(\A)\setminus L(\A')=\emptyset$, and if $\P(L(\A))=0$, then for every DBW $\A'$ embodied in $\A$, we have that $\P(L(\A)\setminus L(\A'))=0$. 

In this section we study approximated determinization by pruning. 
We first show that, unsurprisingly, almost-DBPness is not a trivial property:

\begin{thm}
\label{nbw not almost}
There is an NWW (and hence, also an NBW and an NCW) that is not almost-DBP.
\end{thm}

\begin{proof}
Consider the NWW $\A_1$ in Figure~\ref{ex1}. It is not hard to see that $L(\A_1)=\{a,b\}^\omega$, and so $\P(L(\A_1))=1$. Moreover, every deterministic pruning of $\A_1$ is such that $q_{rej}$ is reachable from all states, which implies that $\{q_{rej}\}$ is the only ergodic SCC of any pruning. Since $\{q_{rej}\}$ is $\alpha$-free, it follows that every deterministic pruning of $\A_1$ recognizes a language of measure zero, and hence $\A_1$ is not almost-DBP. As an example, consider the deterministic pruning $\A'_1$ described on the right hand side of Figure~\ref{ex1}. The only ergodic SCC of $\A'_1$ is $\alpha$-free, and as such $\P(L(\A'_1))=0$.
\end{proof} 
\begin{figure}[ht]
	\centering
	\includegraphics[width=\columnwidth]{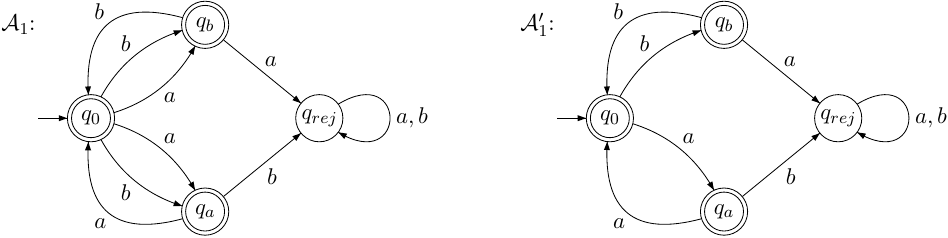}
	\caption{An NWW that is not almost-DBP.} 
	\label{ex1}
\end{figure}

We continue and study the almost-DBPness of HD and SD automata. We show that while for \buchi (and hence also weak) automata, semantic determinism implies almost-DBPness, thus every SD-NBW is almost-DBP, for co-\buchi automata semantic determinism is not enough, and we need HDness. Thus, there is an SD-NCW that is not almost-DBP, yet all HD-NCWs are almost-DBP.

We start with the positive result about \buchi automata. Consider an NBW $\A=\zug{\Sigma,Q,q_0,\delta,\alpha}$. We define a {\em simple stochastic \buchi game}  $\mathcal{G}_\A$ as follows.\footnote{In \cite{CJH03} these games are called simple $1\frac{1}{2}$ -player games with \buchi winning objectives and almost-sure winning criterion.} The game is played between Random and Eve. The set of positions of Random are $Q$, and the set of positions of Eve are $Q\times \Sigma$. The game starts from position $q_0$. A round in the game starts at some position $q\in Q$ and proceeds as follows.

\begin{enumerate}
	\item Random picks a letter $\sigma\in\Sigma$ uniformly, and the game moves to position $(q,\sigma)$.
	\item Eve picks a transition $(q,\sigma,p)\in \Delta$, and the game moves to position $p$.
\end{enumerate} 
 
\noindent 
A probabilistic strategy for Eve is $f:(Q\times\Sigma)^+\rightarrow [0,1]^Q$, where for all histories $x\in (Q\times\Sigma)^*$ and positions of Eve $(q,\sigma)\in Q\times \Sigma$, the function $d=f(x\cdot (q,\sigma)):Q\rightarrow[0,1]$, is a distribution on $Q$ such that $d(p)\neq 0$ implies that $p\in \delta(q,\sigma)$. As usual, we say that a strategy $f$ is {\em memoryless\/}, if it depends only on the current position, thus for all histories $x,y\in (Q\times \Sigma)^*$ and positions of Eve $(q,\sigma)\in Q\times \Sigma$, it holds that $f(x\cdot (q,\sigma))=f(y\cdot (q,\sigma))$. A strategy for Eve is {\em pure\/} if for all histories $x\in (Q\times \Sigma)^*$ and positions of Eve $(q,\sigma)\in Q\times \Sigma$, there is a position $p\in \delta(q,\sigma)$ such that $f(x\cdot(q,\sigma))(p)=1$. When Eve plays according to a strategy $f$, the outcome of the game can be viewed as a run $r_f=q_f^0,q_f^1,q_f^2,\ldots$ in $\A$, over a random word $w_f\in \Sigma^\omega$. (The word that is generated in a play is independent of the strategy of Eve, but we use the notion $w_f$ to emphasize that we are considering the word that is generated in a play where Eve plays according to $f$). 

Let $Q_{rej}=\{q\in Q : \P(L(\A^q))=0\}$.
The outcome $r_f$ of the game is {\em winning\/} for Eve iff $r_f$ is accepting, or $r_f$ visits $Q_{rej}$. Note also that for all positions $q\in Q_{rej}$ and $p\in Q$, if $p$ is reachable from $q$, then $p\in Q_{rej}$. Hence, the winning condition can be defined by the \buchi objective $\alpha\cup Q_{rej}$. Note that $r_f$ is winning for Eve iff $\inf(r_f)\subseteq Q_{rej}$ or $\inf(r_f)\cap \alpha \neq \emptyset$. 
We say that $f$ is an {\em almost-sure winning\/} strategy, if $r_f$ is winning for Eve with probability $1$, and Eve {\em almost-sure wins} in $\G_\A$ if she has an almost-sure winning strategy.

\begin{thm}
\label{sd-nbw almost}
All SD-NBWs are almost-DBP.
\end{thm}

\begin{proof}
Let $\A$ be an SD-NBW, and consider the simple stochastic \buchi game  $\mathcal{G}_\A$.
We first show that Eve has a probabilistic strategy to win $\G_\A$ with probability $1$, even without assuming that $\A$ is semantically deterministic. Consider the probabilistic strategy $g$ where from $(q,\sigma)$, Eve picks one of the $\sigma$-successors of $q$ uniformly at random. Note that this strategy is memoryless, and hence the outcome of the game can be thought as a random walk in $\A$ that starts at $q_0$ and gives positive probabilities to all transitions. Thus, with probability $1$, the run $r_g$ is going to reach an ergodic SCC of $\A$ and visit all its states.  If the run $r_g$ reaches an $\alpha$-free ergodic SCC, then $\inf(r_g)\subseteq Q_{rej}$, and hence $r_g$ is then winning for Eve. Otherwise, $r_g$ reaches a non $\alpha$-free ergodic SCC, and with probability $1$, it visits  all the states in that SCC. Thus, with probability $1$, we have $\inf(r_g)\cap \alpha\neq 0$, and $r_g$ is winning for Eve. Overall, Eve wins $\G_\A$ with probability $1$ when playing according to $g$.

For a strategy $f$ for Eve, we say that {\em $r_f$ is correct} if $w_f\in L(\A)$ implies that $r_f$ is accepting. Note that $w_f\notin L(\A)$ always implies that $r_f$ is rejecting. We show that if $\A$ is SD, then for every almost-sure winning strategy $f$ for Eve, we have that $r_f$ is correct with probability $1$. 
Consider an almost-sure winning strategy $f$ for Eve. Since $f$ is almost-sure winning for Eve, we have that with probability 1, either $r_f$ is an accepting run or $\inf(r_f)\subseteq Q_{rej}$. Thus, it is sufficient to prove that if $\inf(r_f)\subseteq Q_{rej}$, then $w_f\in L(\A)$ with probability $0$, since otherwise, almost surely $r_f$ is accepting, and in particular is correct. 

By semantic determinism, for all $i\geq 0$, it holds that $w_f\in L(\A)$ iff $w_f[i+1,\infty]\in L(\A^{q_f^i})$. Moreover, the 
suffix $w_f[i+1,\infty]$ is independent of $q^i_f$, and hence for all $q\in Q$ and $i\geq 0$, the event $w_f[i+1,\infty]\in L(\A^{q})$ is independent of $q^i_f$.
Thus, for all $q\in Q$ and $i\geq 0$, it holds that $\P(w_f\in L(\A)\: | \: q_f^i=q)=\P(w_f[i+1,\infty]\in L(\A^q))=\P(L(\A^q))$. Hence, by the definition of $Q_{rej}$, and by the fact that $Q_{rej}$ is finite, we have that $\P(w_f\in L(\A)\: | \: q_f^i\in Q_{rej})=0$ for all $i\geq 0$, and so $\P(w_f\in L(\A)\: | \: r_f\text{ visits }Q_{rej})=0$. Overall, we showed that if $f$ is winning for Eve and $\A$ is SD, then $\P(r_f\text{ is correct})=1$.

Hence, by pure memoryless determinacy of simple stochastic parity games~\cite{CJH03}, we may consider a pure memoryless winning strategy $f$ for Eve in $\G_\A$, and by the above we know that a random walk in the corresponding pruned DBW is correct with probability 1. Formally, since $f$ is pure memoryless, it induces a pruning of $\A$. We denote this pruning by $\A^f$, and think of $r_f$ as a random walk in $\A^f$. As we saw, since $f$ is a winning strategy and $\A$ is SD, we have that $r_f$ is correct with probability 1. Note that $\P(L(\A)\setminus L(\A^f))$, is precisely the probability that a random word $w_f$ is in $L(\A)$ but not accepted by $\A^f$. Namely, the probability that $r_f$ is not correct. Hence, $\P(L(\A)\setminus L(\A^f))=0$, and $\A$ is almost-DBP.
\end{proof}

We continue to co-\buchi automata and show that unlike the case of \buchi, here semantic determinism does not imply almost-DBPness. Note that the NWW in Figure~\ref{ex1} is not SD.

\begin{thm}
\label{sdncw not almost}
There is an SD-NCW that is not almost-DBP.
\end{thm}

\begin{proof}
Consider the NCW $\A_2$ in Figure~\ref{ex2}. 
It is not hard to see that $L(\A_2)=\{a,b\}^\omega$, and hence $\P(L(\A_2))=1$. In fact all the states $q$ of $\A_2$ have $L(\A_2^q)= \{a,b\}^\omega$, and so it is semantically deterministic.
\begin{figure}[ht]
		\centering
		\includegraphics[height=4cm]{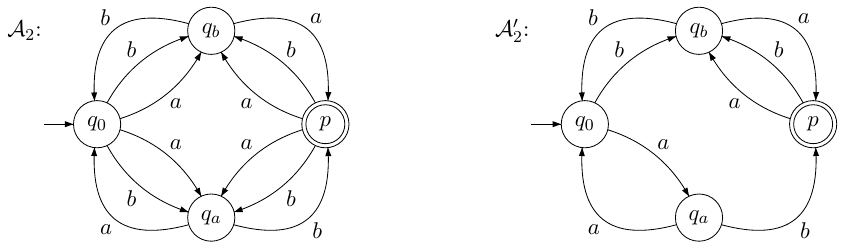}
		\caption{An SD-NCW that is not almost-DBP.} 
		\label{ex2}
	\end{figure}

Moreover, every deterministic pruning of $\A_2$ is strongly connected and not $\alpha$-free. It follows that any deterministic pruning of $\A_2$ recognizes a language of measure zero, and hence $\A_2$ is not almost-DBP.
As an example, consider the deterministic pruning $\A'_2$ described on the right hand side of Figure~\ref{ex2}. It is easy to see that $\A'_2$ is strongly connected and not $\alpha$-free, and as such, $\P(L(\A'_2))=0$.
Thus, $\A_2$ is not almost-DBP.
\end{proof}

Consider a language $L\subseteq \Sigma^\omega$ of infinite words. We say that a finite word $x\in \Sigma^*$ is a {\em good prefix for $L$\/} if $x\cdot \Sigma^\omega \subseteq L$. Then, $L$ is a {\em co-safety\/} language if every word in $L$ has a good prefix \cite{AS87}. Let $\cs(L)=\{ x\cdot w\in \Sigma^\omega : \text{$x$ is a good prefix of }L\}$. Clearly, $\cs(L)\subseteq L$. The other direction is not necessarily true. For example, if $L \subseteq \{a,b\}^\omega$ is the set of all words with infinitely many $a$'s, then $\cs(L)=\emptyset$. In fact, $\cs(L)=L$ iff $L$ is a co-safety language. As we show now, when $L$ is NCW-recognizable, we can relate  $L$ and $\cs(L)$ as follows. 

\begin{lem}\label{good prefix is 1-apprx of NCW}	If $L$ is NCW-recognizable, then $\P(L\setminus \cs(L))=0$.
\end{lem}

\begin{proof}
Consider an NCW-recognizable language $L$. Since NCW=DCW, there is a DCW $\D$ that recognizes $L$. Assume without loss of generality that $\D$ has a single state $q$ with $L(\D^q) = \Sigma^\omega$, in particular, $C = \{q\}$ is the only ergodic $\alpha$-free SCC of $\D$. Then, for every word $w \in \Sigma^\omega$, we have that $w \in \cs(L)$ iff the run $r$ of $\D$ on $w$ reaches $C$. Hence, the probability that $w \in L\setminus \cs(L)$ equals the probability that $\inf(r)$ is $\alpha$-free but is not an ergodic SCC of $\D$. Since the later happens with probability 0, we have that  $\P(L\setminus \cs(L))=0$.
\end{proof}

By Lemma~\ref{good prefix is 1-apprx of NCW}, pruning an NCW in a way that would make it recognize $\cs(L(\A))$ results in a DCW that approximates $\A$, and thus witnesses that $\A$ is almost-DBP. We now show that for HD-NCWs, such a pruning is possible, and conclude that HD-NCWs are almost-DBP. 

\begin{thm}\label{HD-ncw almost}All HD-NCWs are almost-DBP.\end{thm}

\begin{proof}Let $\A=\zug{\Sigma,Q,q_0,\delta,\alpha}$ be an HD-NCW. Consider the NCW $\A'=\zug{\Sigma,Q,q_0,\delta,\alpha'}$, where $\alpha'= \alpha\cup \{q\in Q : L(\A^q)\neq \Sigma^\omega\}$. We prove that $\A'$ is an HD-NCW with $L(\A')=\cs(L(\A))$. 
Consider a word $w\in L(\A')$, and let $r=q_0,q_1,q_2,\ldots$ be an accepting run of $\A'$ on $w$. There exists a prefix $x\in \Sigma^*$  of $w$ such that $r$ reaches some $q\notin \alpha'$ when reading $x$. Hence $L(\A^q)=\Sigma^\omega$, and so $x\cdot \Sigma^\omega\subseteq L(\A)$. That is, $x$ is a good prefix and $w\in \cs(L(\A))$. Thus, $L(\A')\subseteq \cs(L(\A))$. 

In order to see that $\A'$ is HD and that $\cs(L(\A))\subseteq L(\A')$, consider an HD strategy $f$ of $\A$. We claim that $f$ is also an HD strategy for $\A'$, and for all words $w\in \cs(L(\A))$, the run $r$ that $f$ generates on $w$ eventually visits only states $q\notin \alpha'$. Since $\cs(L(\A))\subseteq L(\A)$, we know that $\inf(r)\cap \alpha=\emptyset$. It is left to prove that $r$ visits only finitely many states $q\in Q$ with $L(\A^q)\neq\Sigma^\omega$. 

Since $w\in \cs(L(\A))$, it has a good prefix $x\in \Sigma^*$. Since $f$ is an HD strategy and $x$ is a good prefix, it must be that $L(\A^{f(x)})=\Sigma^\omega$. Also, since all the extensions of $x$ are also good prefixes, the above holds also for all states that $f$ visits after $f(x)$. Hence, the run $r$ visits states $q$ with $L(\A^q)\neq \Sigma^\omega$ only during the finite prefix $x$. Thus, $\inf(r)\cap \alpha'=\emptyset$, and we are done.

So, $\A'$ is an HD-NCW with $L(\A')=\cs(L(\A))$. Since $\cs(L(\A))$ is co-safe, it is DWW-recognizable \cite{Sis94}. By \cite{BKS17}, HD-NCWs whose language is DWW-realizable are DBP. Let $\delta'$ be the restriction of $\delta$ to a deterministic transition function such that $\D'=\zug{\Sigma,Q,q_0,\delta',\alpha'}$ is a DCW with $L(\D')=L(\A')=\cs(L(\A))$. Consider now the DCW $\D=\zug{\Sigma,Q,q_0,\delta',\alpha}$ that is obtained from $\D'$ by replacing $\alpha'$ with $\alpha$.

 It is clear that $\D$ is a pruning of $\A$. Note that, $\alpha\subseteq \alpha'$, and hence $\cs(L(\A))=L(\D')\subseteq L(\D)$. That is, $\D$ is a pruning of $\A$ that approximates $L(\A)$ up to a negligible set, and $\A$ is almost-DBP. 
\end{proof}

\section{Nondeterminism and Reasoning about MDPs}

A probabilistic open system can be modeled by a \emph{Markov decision process} (MDP, for short). As we define formally below, an MDP describes the behavior of the system when it interacts with its environment. The behavior is probabilistic: each input provided by the environment induces a probability on the reaction of the system. Reasoning about the behavior of an MDP can proceed by taking its product with a deterministic automaton that describes a desired behavior for the MDP. 
Indeed, this product contains information on the probability that the interaction results in a behavior accepted by the automaton. When the automaton is nondeterministic, this information is lost, and thus reasoning about MDPs involves determinization. An automaton is said to be {\em good-for-MDPs} (GFM), if its product with MDPs maintains the probability of acceptance, and can therefore replace deterministic automata when reasoning about stochastic behaviors \cite{HPSSTWW20,STZ22}. In Section~\ref{sec prob}, we related semantic determinism with almost-DBPness. 
In this section we relate semantic determinism with GFMness. 
The terminology used in this section is based on \cite{HPSSTWW20}.

\subsection{Good-for-MDP automata}

We first define MDPs and good-for-MDP automata. 
A $\Sigma$-labeled \emph{Markov decision process} (MDP) is $\M =\zug{S,s_0,(A_s)_{s\in S},\Pr_{\M},\Sigma,\tau}$, where $S$ is a finite set of states, $s_0\in S$ is an initial state, and for every $s\in S$, the set $A_s$ is a finite set of actions that are available in $s$. Let $A=\bigcup_{s\in S}A_s$ be the set of all actions. Then, $\Pr_\M:S\times A\times S\nrightarrow [0,1]$ is a (partial) stochastic transition function: for every two states $s,s'\in S$ and action $c\in A_s$, we have that $\Pr_\M(s,c,s')$ is the probability of moving from $s$ to $s'$ when action $c$ is taken. Accordingly, for every $s\in S$ and $c\in A_s$, we have $\Sigma_{s'\in S}\Pr_\M(s,c,s') = 1$. Finally, $\Sigma$ is a finite alphabet, and $\tau:S\rightarrow \Sigma$ is a labeling function on the states.

Consider a nondeterministic automaton $\N$ with alphabet $\Sigma$, and a $\Sigma$-labeled MDP $\M$, and consider the following game between a player and nature. The positions of the game are pairs $\zug{s,q}$, where $s\in S$ is a state of $\M$, and $q\in Q$ is a state of $\N$. The game starts by the player picking an action $c\in A_{s_0}$ available from the initial state of $\M$, and the game moves to $\zug{s,q_0}$ with probability $\Pr_\M(s_0,c,s)$. Then, at each round from position $\zug{s,q}$, the player chooses an action $c\in A_s$ and a state $q'\in \delta(q,\tau(s))$, and then nature chooses a state $s'\in S$ with probability $\Pr_\M(s,c,s')$, and the game moves to $\zug{s',q'}$. 
The objective of the player is to on the one hand maximize the probability of generating a word in $L(\N)$ and on the other hand, to generate a rejecting run on words in $L(\N)$ with zero probability.
When the automaton $\N$ is deterministic, the player has to only pick actions, as the state $q'\in \delta(q,\tau(s))$ is unique. When the automaton is nondeterministic, the player may fail to generate an accepting run of $\N$ even when the actions he provides induce a word in $L(\N)$. Indeed, each nondeterministic choice may lead to accepting only a subset of the words in $L(\N)$. 
Intuitively, $\N$ is good-for-MDP, if for every MDP $\M$, the nondeterminism in $\N$ does not prevent the player from maximizing the probability of generating a word in $L(\N)$ while generating a rejecting run on words in $L(\N)$ with zero probability. We now formalize this intuition.

A {\em run\/} of $\M$ is an infinite sequence of states $r=\zug{s_0,s_1,s_2,\ldots}\in S^\omega$, such that for all $i\geq 0$ there exists $c\in A_{s_i}$ with $\Pr_\M(s_i,c,s_{i+1}) > 0$. A finite run is a prefix of such a sequence. Let $\Omega(\M)\subseteq S^\omega$ denote the set of runs of $\M$ from the initial state $s_0$, and let $\Pi(\M)\subseteq S^*$ denote the set of finite such runs. 
For a run $r\in \Omega(\M)$, we define the corresponding labeled run as $\tau(r) = \tau(s_1),\tau(s_2),\tau(s_3),\ldots\in \Sigma^\omega$. Note that the label of the first state $s_0$ of a run is not part of $\tau(r)$. 

Consider the set $A$ of actions. Let $\dist(A)$ be the set of distributions over $A$. A distribution $d\in \dist(A)$ is a \emph{point distribution} if there exists $c\in A$ such that $d(c)=1$. A \emph{strategy} for an MDP $\M$ is a function $\mu : \Pi(\M)\rightarrow \dist(A)$, which suggests to the player a distribution on the available actions given the history of the game so far. The strategy should suggest an available action, thus $\mu(s_0,s_1,s_2,\ldots,s_k)(c)>0$ implies that $c\in A_{s_k}$. Let $\Pi(\M,\mu)$ denote the subset of $\Pi(\M)$ that includes exactly all finite runs that agree with the strategy $\mu$. Formally, we have that $\zug{s_0}\in \Pi(\M,\mu)$, and for all $p'=\zug{s_0,s_1,\ldots,s_k}\in \Pi(\M,\mu)$, and $s_{k+1}\in S$, we have that $p=p'\cdot s_{k+1}\in\Pi(\M,\mu)$ iff there exists $c_{k+1}\in A_{s_k}$ such that $\mu(p')(c_{k+1})>0$ and $\Pr_{\M}(s_{k},c_{k+1},s_{k+1})>0$. Thus, a run in $\Pi(\M,\mu)$ can only be extended using actions that have a positive probability according to $\mu$.
Let $\Omega(\M,\mu)$ denote the subset of $\Omega(\M)$ that includes exactly all runs $r$ such that for every $p\in \Pi(\M)$ that is a prefix of $r$, it holds that $p\in \Pi(\M,\mu)$. Let $\F(\M)$ be the set of all strategies.

We say that a strategy $\mu$ is \emph{pure} if for all $p \in \Pi(\M)$, the distribution $\mu(p)$ is a point distribution, and is \emph{mixed} otherwise.
We say that $\mu$ has \emph{finite-memory}, if there exists a finite set $M$ of memories, with a distinguished initial memory $m_0\in M$, and a memory update function $\eta:S\times M\to M$ such that $\mu$ only depends on the current memory and MDP state. Formally, $\eta$ is extended to finite sequences in $S^*$ as follows: $\eta^*: S^* \to M$ is such that for every finite sequence $p\in S^*$ and MDP state $s\in S$, we have that $\eta^*(\epsilon) = m_0$ and $\eta^*(p\cdot s) = \eta(s,\eta^*(p))$. 
Then, for all $s\in S$ and $p\cdot s,p'\cdot s\in \Pi(\M,\mu)$, if $\eta^*(p\cdot s)=\eta^*(p'\cdot s)$ then $\mu(p\cdot s)=\mu(p'\cdot s)$. Equivalently, $\mu$ can be thought as a function $\mu:S\times M\rightarrow \dist(A)\times M$, that for a given the current MDP state $s$ and memory $m$, outputs a distribution on the available actions from $s$ and the next memory state to move to from $m$. 

Indeed, we show that every finite memory strategy $\mu$ with memory set $M$ can be represented as a strategy of the form $\mu':S\times M \rightarrow \dist(A)\times M$. We say that $\zug{s,m}\in S\times M$ is \emph{feasible} by $\mu$ if there is a finite run $p_{m,s}$ such that $p_{m,s}\cdot s\in \Pi(\M,\mu)$ and $\eta^*(p_{m,s}\cdot s)=m$. In other words, there is a positive probability that when playing according to $\mu$ we reach the memory state $m$ when being in state $s$ in the MDP. We can now define $\mu': S\times M\nrightarrow \dist(A)\times M$ as a partial function on all feasible $\zug{s, m}\in S\times M$ by $\mu'(s,m)=\zug{\mu(p_{m,s}\cdot s),\eta(s,m)}$.  Observe that $\mu(p_{m,s}\cdot s)$ is independent of the choice of witness $p_{m,s}$. That is, for all $p\cdot s\in \Pi(\M,\mu)$ with $\eta^*(p\cdot s)=m$ we have that $\mu(p\cdot s)=\mu(p_{m,s}\cdot s)$. Thus, if $\mu'(s,m)=\zug{d,m'}$, then $\mu(p\cdot s)=d$ for all $p\cdot s\in \Pi(\M,\mu)$ such that $\eta^*(p\cdot s)=m$.

We say that $\mu$ is \emph{memoryless} if it has a finite memory $M$ with $|M|=1$. Equivalently, if for all $p,p'\in \Pi(\M,\mu)$ and $s\in S$ such that $p\cdot s,p'\cdot s\in \Pi(\M,\mu)$, it holds that $\mu(p\cdot s)=\mu(p'\cdot s)$. Namely, if $\mu$ only depends on the current state, and not on the history beforehand.

A \emph{Markov Chain} (MC, for short) is an MDP with a dummy set of actions, $A_s=\set{\ast}$ for all $s\in S$. Formally, an MC is $\C=\zug{S,s_0,\Pr_\C,\Sigma,\tau}$, where $S$ is the set of states, $s_0\in S$ is the initial state, $\Pr_\C:S\times S\rightarrow[0,1]$ is a stochastic transition function: for every $s\in S$ it holds that $\sum_{s'\in S}\Pr_\C(s,s')=1$. A \emph{random process} of $\C$, is an infinite sequence of random variables $S_0,S_1,S_2,\ldots$ with values in $S$, such that $S_0=s_0$ with probability 1, and for all $s_1,s_2,\ldots,s_{i+1}\in S$, $i\geq 0$, it holds that $\P_\C(S_{i+1}=s_{i+1}|S_0 = s_0,\ldots,S_i=s_i)=\Pr_\C(s_i,s_{i+1})$.  

Given a strategy $\mu$, we can obtain from $\M$ and $\mu$ the MC $\C(\M,\mu) = \zug{\Pi(\M,\mu),s_0,\Pr_\M^\mu}$ in which the choice of actions is resolved according to $\mu$. Formally, for all $p\in \Pi(\M,\mu)$ and $s_1,s_2\in S$ such that both $p_1=p\cdot s_1$ and $p_2=p_1\cdot s_2$ belong to $\Pi(\M,\mu)$, then $\Pr_\M^\mu(p_1,p_2) = \sum_{c\in A_{s_1}}\mu(p_1)(c)\cdot \Pr_\M(s_1,c,s_2)$. 
It is clear that when $\mu$ is memoryless, that is, $\mu$ only depends on the last position in $p_1$, then ${\Pr}_\M^\mu$ only depends on $s_1$ and $s_2$, and hence the projection of the chain $\C(\M,\mu)$ onto the last position in a path is a finite MC. Formally, let $\iota:\Pi(\M,\mu)\rightarrow S$ be the function that maps a finite path to its last state and let $d_s\in \dist(\A)$ be the unique distribution such that $\mu(p)=d_s$ for all $p\in \Pi(\M,\mu)$ with $\iota(p)=s$. Note that if no such $p\in \Pi(\M,\mu)$ with $\iota(p)=s$ exists, then we arbitrarily take $d_s$ to be some distribution in $A_s$. 
Let ${\Pr'}_{\M}^\mu:S\times S\rightarrow [0,1]$ be the stochastic transition function that is defined by ${\Pr'}_\M^\mu(s,s')=\sum_{c\in A_{s}}d_s(c)\cdot \Pr_\M(s,c,s')$. Then, if $p_0,p_1,p_2,\ldots \in (\Pi(\M,\mu))^\omega$ is an infinite sequence of increasing paths that are sampled according to $\C(\M,\mu)$, then it is not hard to see that the sequence $\iota(p_0),\iota(p_1),\iota(p_2),\ldots \in S^\omega$ satisfies $\P(\iota(p_{i+1})=s'|\iota(p_i)=s)={\Pr'}_\M^\mu(s,s')$ for all $s',s\in S$. In other words, the sequence $\iota(p_0),\iota(p_1),\iota(p_2),\ldots$ is sampled according to the finite markov chain $\zug{S,s_0,{\Pr'}_\M^\mu}$.

Let $\M=\zug{S,s_0,(A_s)_{s\in S},\Pr_\M,\Sigma,\tau}$ be an MDP and $\N=\zug{\Sigma,Q,q_0,\delta,\alpha}$ be a nondeterministic automaton. When the player takes actions in $\M$ the infinite word $\tau(r)\in \Sigma^\omega$ is generated. Assume that, simultaneously, the player also generates a run of $\N$. This process can be modeled by an MDP that is a product of $\M$ with $\N$, and in each step the player takes an action of $\M$ and chooses a successor state in $\N$ that extends the current run on the new letter drawn by $\M$. Formally, we define the MDP $\M\times \N=\zug{S',\zug{s_0,\bot},A',\Pr_{\M\times \N},\Sigma,\gamma}$ that is defined as follows. The set of states of $\M\times \N$ is $S'=(S\times Q)\cup \set{\zug{s_0,\bot}}$, for $\bot\notin Q$. The set of actions $A'=A_{}\zug{s_0,\bot}\cup \bigcup_{\zug{s,q}\in S\times Q}(A_{\zug{s,q}})$ is defined as follows. For all $\zug{s,q}\in S\times Q$, the set of actions $A_{\zug{s,q}}$ available from $\zug{s,q}$ consists of all pairs $\zug{c,q'}\in A_s\times \delta(q,\tau(s))$. In addition, the actions available from $\zug{s_0,\bot}$ are $A_{\zug{s_0,\bot}}=A_{s_0}\times \set{q_0}$. Then, for all $\zug{s,\ast}\in S'$ and $\zug{c,q'}\in A_{\zug{s,\ast}}$, we have $\Pr_{\M\times \N}(\zug{s,\ast},\zug{c,q'},\zug{s',q'})=\Pr_{\M}(s,c,s')$ and $\gamma(s,\ast)=\tau(s)$.
Given a finite or infinite run $r\in \Pi(\M\times \N)\cup \Omega(\M\times \N)$, $r=\zug{\zug{s_0,\bot},\zug{s_1,q_0},\zug{s_2,q_1},\ldots}$, we define $\pi_\N(r)=\zug{q_0,q_1,q_2,\ldots}$, and $\pi_\M(r)=\zug{s_0,s_1,s_2,\ldots}$ to be the $\N$-run and $\M$-run that correspond to $r$ respectively. Conversely, for $r_\M=\zug{s_0,s_1,s_2,\ldots}\in S^*\cup S^\omega$ and $r_\N=\zug{q_0,q_1,q_2,\ldots}\in (Q_\N)^*\cup (Q_\N)^\omega$ 
of appropriate lengths, let $r_\M\oplus r_\N=\zug{\zug{s_0,\bot},\zug{s_1,q_0},\zug{s_2,q_1},\ldots}$. Notice that for all $r_\M\oplus r_\N\in \Pi(\M\times \N)\cup \Omega(\M\times \N)$, we have $\pi_\N(r_\M\oplus r_\N)=r_\N$ and $\pi_\M(r_\M\oplus r_\N)=r_\M$.

We can now formally define the notion an automaton being \emph{good-for-MDP}. We use the formalism of \cite{HPSSTWW20}.
The \emph{semantic satisfaction probability\/} of a strategy $\mu\in \F(\M)$ with respect to $\N$, denoted $\mathrm{PSem}^\N_\M(\mu)$, is defined as:
\[
\mathrm{PSem}^\N_\M(\mu)={\Pr}_\M^\mu(r\in \Omega(\M,\mu) : \tau(r)\in L(\N)).
\]
That is, $\mathrm{PSem}^\N_\M(\mu)$ is the probability that when the player plays with $\mu$, the word $\tau(r)$ generated by $\M$ is in $L(\N)$. 
Then, the \emph{semantic satisfaction probability} of an MDP $\M$ with respect to $\N$, denoted $\mathrm{PSem}^\N_\M$, is defined by
\[
\mathrm{PSem}^\N_\M=\sup_{\mu\in \F(\M)}\mathrm{PSem}^\N_\M(\mu).
\]
That is, $\mathrm{PSem}^\N_\M$ is the least upper bound on how good the player can do with respect to the objective of generating a word in $L(\N)$. 

Similarly, the \emph{syntactic satisfaction probability} of $\mu\in \F(\M\times \N)$ and $\M$, denoted $\mathrm{PSyn}^\N_\M(\mu)$ and $\mathrm{PSyn}^\N_{\M}$ respectively, are defined by
\[
\mathrm{PSyn}^\N_\M(\mu)={\Pr}_{\M\times \N}^\mu(r\in \Omega(\M\times \N,\mu) : \pi_\N(r) \mbox{ is accepting}),
\]
and
\[
\mathrm{PSyn}^\N_{\M}=\sup_{\mu\in \F(\M\times \N)}\mathrm{PSyn}^\N_\M(\mu).
\]
Maximizing the syntactic satisfaction can be expressed as a simple stochastic parity game that is played on top of the MDP $\M\times \N$, where the winning objective is inherited from $\N$. In ~\cite{CH03}, the authors prove that simple stochastic parity games admit optimal pure memoryless strategies. Thus, the supremum in the definition of $\mathrm{PSyn}^\N_{\M}$ is attained by a pure memoryless strategy. Note that $\mathrm{PSyn}^\N_{\M}$ is evaluated with respect to the MDP $\M\times \N$, and thus by pure memoryless we mean that there is a strategy $\mu^*_\N:S\times Q\cup\set{\zug{s_0,\bot}}\rightarrow A\times Q$ such that $\mathrm{PSyn}^\N_{\M}(\mu^*_\N)=\mathrm{PSyn}^\N_{\M}$.

Similarly, the semantic satisfaction can be expressed as a simple stochastic parity game that augments $\M\times\N$ with a deterministic automaton $\D$ for $L(\N)$. Thus, the supremum in the definition of $\mathrm{PSem}^\N_{\M}$ is attained by a pure finite-memory strategy. Indeed, the considerations are as in the syntactic case, except that here, the current state of the deterministic automaton $\D$ is used as a memory for the player.

In general, $\mathrm{PSyn}^\N_\M \leq \mathrm{PSem}^\N_\M$. In the syntactic case, the player not only aims to generate a word in $L(\N)$, but also to simultaneously generate an accepting run of $\N$ on it. Therefore, any strategy for the syntactic case can be used for the semantic case by ignoring the component of $\N$, resulting in a potentially higher semantic satisfaction. Equality, however, is not guaranteed. Indeed, the same way not all automata are history deterministic, it may not be possible to resolve the nondeterminism in $\N$ in a way that maximizes the probability of generating a word in $L(\N)$ (thus reaching a maximal semantic satisfaction) while generating rejecting runs on words in $L(\N)$ with zero probability (and thus reaching an equal syntactic satisfaction). When the equality holds for all MDP $\M$, we say that $\N$ is \emph{good-for-MDP} (GFM for short). Note that when $\N$ is GFM, it holds that for all MDP $\M=\zug{S,s_0,(A_s)_{s\in S},\Pr_{\M},\Sigma,\tau}$, there exists a pure-memroyless strategy for $\M\times \N$, $\mu^*_\N:S\times Q\cup\set{\zug{s_0,\bot}}\rightarrow A\times Q$, such that $\mathrm{PSyn}^\N_{\M}=\mathrm{PSyn}^\N_{\M}(\mu^*_\N)$ and $\mu^*_\N$ generates a rejecting run on words in $L(\N)$ with zero probability. Note that $\mu^*_\N$ can be viewed as a pure finite-memory strategy for $\M$ for which $\mathrm{PSem}^\N_{\M}=\mathrm{PSem}^\N_{\M}(\mu^*_\N)$.

\subsection{The nondeterminism hierarchy and GFMness}

In this section we locate GFMness in the nondeterminism hierarchy. For the semantic hierarchy, it is proved in~\cite{HPSSTWW20} that GFM-NBWs are as expressive as NBWs. More specifically, every NBW can be translated into an equivalent GFM-NBW. Unlike the case of co-B\"uchi automata, where nondeterministic and deterministic automata have the same expressive power, making the expressiveness equivalence of GFM-NCWs and NCW straightforward,  NBWs are more expressive than DBWs, making the result about GFM-NBWs interesting. Another aspect that has been studied is deciding GFMness~\cite{STZ22}. In ~\cite{STZ22},  the authors prove an exponential-time upper bound as well as  a polynomial-space lower  bound for the problem of deciding whether an automaton is GFM-NBW. 

We study the syntactic hierarchy, we first point out that every HD automaton is GFM~\cite{KMBK14}. 
Then, by relating GFMness and almost-DBPness, we show that SD-NBWs are GFM, yet SD-NCWs need not be GFM. Thus, with respect to the syntactic hierarchy, GFM-NBWs are between SD and nondeterministic B\"uchi automata, while GFM-NCWs are incomparable to SD-NCWs.
 
\begin{thmC}[\cite{KMBK14}]
\label{hd implies gfm}
	Every HD automaton is GFM.
\end{thmC}

\begin{proof}
	The theorem is proved in~\cite{KMBK14}. We give a short proof for completeness.
Consider an HD automaton $\N$. Let $f$ be an HD strategy for $\N$. Then, every strategy $\mu\in\F(\M)$ can be augmented with $f$ to obtain a strategy $\mu\times f\in \F(\M\times \N)$ such that $\mathrm{PSem}^\N_{\M}(\mu)=\mathrm{PSem}^\N_{\M\times \N}(\mu\times f)=\mathrm{PSyn}^\N_\M(\mu\times f)$. Thus, $\mathrm{PSem}^\N_\M\leq \mathrm{PSyn}^\N_\M$, and so $\N$ is GFM.
\end{proof}

\begin{thm}
	\label{gfm-nbw almost}
	Every GFM automaton is almost-DBP.
\end{thm}

\begin{proof}
	Consider an $\omega$-regular automaton $\N$ over an alphabet $\Sigma$, and consider the MDP $\M=\zug{\Sigma,\sigma_0,(A_\sigma)_{\sigma\in\Sigma},\Pr_{\M},\mathsf{id}}$, where $\sigma_0\in \Sigma$ is an arbitrary initial state, $A_\sigma=\set{\ast}$ is a dummy set of actions for each letter $\sigma\in \Sigma$, and $\Pr_{\M}(\sigma,\ast,\sigma')=|\Sigma|^{-1}$, for all $\sigma,\sigma'\in\Sigma$. By \cite{CJH03}, there is a pure memoryless strategy $\mu^*:S\times Q_\N\cup\set{\zug{s_0,\bot}}\rightarrow A\times Q_\N$ that attains $\mathrm{PSyn}^\N_\M$. Observe that since $\Pr_{\M}(\sigma,\ast,\cdot)=(\sigma'\mapsto |\Sigma|^{-1})$ is the uniform distribution on $\Sigma$ for all $\sigma\in \Sigma$, then $\M$ is essentially an MC that generates a random word $w=\sigma_1,\sigma_2,\sigma_3,\ldots\in \Sigma^\omega$, such that all its letters are independent and distributed uniformly. Thus, the membership $w\in L(\N)$, where $w$ is the word generated by $\C(\M,\mu')$ (or by $\C(\M\times \N,\mu')$) for $\mu'\in \F(\M)$ (respectively $\mu'\in \F(\M\times \N)$), is independent of the strategy $\mu'$. In particular $\mathrm{PSem}^\N_\M$ equals the probability that a random word $w$ is in $L(\N)$. 
	
	Note that it must be that $\mu^*(\sigma_0,\bot)=\zug{\ast,q_0}$; i.e., there is no choice at the initial state of $\M\times \N$. Also, as $A_\sigma=\set{\ast}$, for all $\sigma\in \Sigma$, we can think of $\mu^*$ as a function $\mu^*:\Sigma\times Q_\N\rightarrow Q_\N$, with $\mu^*(\sigma,q)\in \delta_\N(q,\sigma)$ for all $\zug{\sigma,q}\in \Sigma\times Q_\N$. Thus, $\mu^*$ induces a pruning $\N_{\mu^*}$ of $\N$, and $\mathrm{PSyn}^\N_\M(\mu^*)$ equals to the probability that a random word is accepted by $\N_{\mu^*}$. 
	
	Now, since $\mu^*$ is an optimal strategy for $\M\times \N$ and $\N$ is GFM, it follows that $\mathrm{PSyn}^\N_\M(\mu^*)=\mathrm{PSyn}^\N_\M=\mathrm{PSem}^\N_\M$. Hence, the probability that a random word $w\in \Sigma^\omega$ is accepted by $\N_{\mu^*}$ equals the probability that $w\in L(\N)$. That is, $\N_{\mu^*}$ is a deterministic pruning of $\N$ that loses only a set of words of measure zero and hence $\N$ is almost-DBP.
\end{proof}

Recall that in Theorem~\ref{sdncw not almost}, we argued that there is an SD-NCW that is not almost-DBP. By Theorem~\ref{gfm-nbw almost}, we can conclude that the SD-NCW described there, is also not GFM. 
Thus, while HDness implies GFMness for all acceptance conditions~\cite{KMBK14}, for co-\buchi automata, SDness need not imply GFMness. 
In the rest of this section we prove that for \buchi automata, SDness does imply GFMness, thus all SD-NBWs are GFM. In particular, Theorem~\ref{gfm-nbw almost} suggests an alternative proof of Theorem~\ref{sd-nbw almost}, which states that all SD-NBWs are almost-DBP.

\begin{prop}\label{sd unif str}
	If $\N$ is an SD-NBW and $\mu\in \F(\M)$ is a finite memory strategy, then there exists $\mu'\in \F(\M\times \N)$ such that $\mathrm{PSem}^\N_\M(\mu)=\mathrm{PSyn}^\N_\M(\mu')$.
\end{prop}

\begin{proof}
	Let $\mu:S\times M\rightarrow \dist(A)\times M$, be a finite memory strategy in $\F(\M)$, and let $\C=\C(\M,\mu)$ be the Markov chain induced when taking actions in $\M$ according to $\mu$. Note that since $\mu$ uses $M$ as a memory structure, we may think of $\C$ as a Markov chain with positions $S\times M$, and a stochastic transition function $\Pr_\C:(S\times M)^2\rightarrow[0,1]$ that is defined by $\Pr_\C(\zug{s,m},\zug{s',m'})=\sum_{c\in A_s}d(c) \cdot \Pr_{\M}(s,c,s')$, when $\mu(s,m)=\zug{d,m'}$, and $\Pr_\C(\zug{s,m},\zug{s',m'})=0$, otherwise. Let $r=\zug{\zug{s_0,m_0},\zug{s_1,m_1},\zug{s_2,m_2},\ldots} \in (S\times M)^\omega$ be a random process sampled according to $\C$. 
	Let $r_\N=\zug{q_0,q_1,q_2,\ldots} \in (Q_\N)^\omega$ be a random run of $\N$ over the word $\tau(r)$ that is sampled when playing in $\M$ with the strategy $\mu$. I.e., given $\zug{s_0,m_0},\zug{s_1,m_1},\ldots,\zug{s_k,m_k}$, for $k\geq 1$, and $q_0,q_1,\ldots,q_{k-1}$, we set $q_k$ to be one of the $\tau(s_k)$-successors of $q_{k-1}$ at random. Notice that the composition $r\oplus r_\N= \zug{\zug{s_0,m_0,q_0},\zug{s_1,m_1,q_1},\ldots}$ is a process of a finite Markov chain $\C_\N$ with a set of positions $S\times M\times Q_\N$. Indeed, the stochastic transition function of $\C_\N$ can be expressed as follows:
	\begin{equation*}
		\Pr_{\C_\N}(\zug{s,m,q},\zug{s',m',q'})=
		\begin{cases}
			|\delta_\N(q,\tau(s'))|^{-1}\cdot \Pr_\C(\zug{s,m},\zug{s',m'})&:q'\in \delta_\N(q,\tau(s'))\\
			0&:\mbox{otherwise}
		\end{cases}
	\end{equation*}  

	Since $\C_\N$ is a finite Markov chain, it holds that with probability $1$, the random process $r\oplus r_\N$ of $\C_\N$ eventually enters an ergodic component and traverses all its states. We say that an ergodic component $C\subseteq S\times M\times Q_\N$ of $\C_\N$ is \emph{$\alpha_\N$-accepting} if the projection of $C$ onto $Q_\N$ intersects $\alpha_\N$. I.e., $C$ is $\alpha_\N$-accepting if there exist $s\in S$, $m\in M$ and $q\in \alpha_\N$ such that $\zug{s, m, q}\in C$. Thus, if $r\oplus r_\N$ eventually enters an $\alpha_\N$-accepting ergodic component $C$, then with probability 1 the process traverses all the positions in $C$ infinitely often, and in particular with probability 1 $r_\N$ is an accepting run of $\N$ over the generated word $\tau(r)$. Since with probability 1 the process of $\C_\N$ eventually enters an ergodic component, it follows that the probability that $\tau(r)\in L(\N)$ but $r_\N$ is rejecting equals the probability that $r\oplus r_\N$ enters an ergodic component that is not $\alpha_\N$-accepting and $\tau(r)\in L(\N)$. We prove that if $r\oplus r_\N$ enters an ergodic component $C$ that is not $\alpha_\N$-accepting then $\tau(r)\notin L(\N)$. Note that we claim that the implication $C$ not $\alpha_\N$-accepting implies that $\tau(r)\notin L(\N)$, is surely, and not only almost-surely. 
	Consider an ergodic component $C$ that is not $\alpha_\N$-accepting, and assume that $r\oplus r_\N$ is such that $\zug{s_k,m_k,q_k}\in C$ for some $k\geq 0$. Assume by way of contradiction that $\tau(r)\in L(\N)$. Since $\N$ is SD, it is possible to extend the finite run $q_0,q_1,\ldots q_k$ of $\N$ over $\tau(s_1),\tau(s_2),\ldots,\tau(s_k)$ into an accepting run of $\N$ over $\tau(r)$. In particular there exist $j>k$ and $q'_{k+1},q'_{k+2},\ldots,q'_j\in (Q_\N)^*$, such that $q'_j\in \alpha_\N$ and $q_0,\ldots,q_k,q'_{k+1},\ldots,q'_j$ is a finite run of $\N$ over $\tau(s_1),\ldots,\tau(s_j)$. By definition of the stochastic transition function of $\C_\N$, it follows that there is a positive probability to move to $\zug{s_{k+1},m_{k+1},q'_{k+1}}$ from $\zug{s_k,m_k,q_k}$. Similarly, it follows by induction that for all $k+1< i\leq j$ there is a positive probability to move to $\zug{s_i,m_i,q'_i}$ from $\zug{s_{i-1},m_{i-1},q'_{i-1}}$. In particular $\zug{s_j,m_j,q'_j}\in C$, in contradiction to the assumption that $C$ is not $\alpha_\N$-accepting.
	
	Let $\mu'\in \F(\M\times \N)$ be the strategy for $\M\times \N$ that resolves the actions of $\M$ according to $\mu$ and in addition generates a random run of $\N$ over $\tau(r)$. It is not hard to see that $\C_\N$ induces the same distribution over infinite paths in $(S\times Q_\N)^\omega$ as the Markov chain $\C(\M\times\N,\mu')$. We proved that there is a zero probability that $\tau(r)\in L(\N)$ and $r_\N$ is rejecting. Thus, $\mathrm{PSyn}^\N_\M(\mu')=\mathrm{PSem}^\N_\M(\mu)$. 
\end{proof}

We can now conclude with the desired connection between SDness and GFMness. 
We first examine whether SDness implies GFMness, and show that the answer depends on the acceptance condition.

\begin{thm}\label{sd-nbw gfm}
	All SD-NBWs are GFM, yet there is an SD-NCW that is not GFM. 
\end{thm}

\begin{proof}
The co-\buchi case follows from Theorems~\ref{sdncw not almost} and~\ref{gfm-nbw almost}. For \buchi automata, consider an SD-NBW $\N=\zug{\Sigma,Q_\N,q_0,\delta_\N,\alpha_\N}$, and a $\Sigma$-labeled MDP $\M$. Let $\mu$ be a finite memory strategy such that $\mathrm{PSem}^\N_\M(\mu)=\mathrm{PSem}^\N_\M$ given by \cite{CJH03}. Then by Proposition~\ref{sd unif str} it holds that there exists $\mu'\in \F(\M\times \N)$ such that $\mathrm{PSyn}^\N_\M(\mu')=\mathrm{PSem}^\N_\M(\mu)$, and we get $\mathrm{PSem}^\N_\M=\mathrm{PSyn}^\N_\M(\mu')\leq \mathrm{PSyn}^\N_\M$. Since $\mathrm{PSyn}^\N_\M\leq \mathrm{PSem}^\N_\M$ holds for all $\N$ and $\M$, we have the equality $\mathrm{PSyn}^\N_\M= \mathrm{PSem}^\N_\M$.
\end{proof}

We continue and check whether GFMness implies SDness, and show that the answer is negative even for weak automata.

\begin{thm}
	There is a GFM-NWW that is not SD.
\end{thm}

\begin{proof}
		Consider the NBW $\N$ described in Figure~\ref{GfmNotSD}. Note that $\N$ is weak, as all cycles are self-loops. Also, $\N$ is not SD, as $a\cdot b^\omega \in L( \N^{q_0}) \setminus L(\N^{q_{acc}})$. 
		Also, it is easy to see that $L(\N) = \{w: \text{$w$ has finitely many $a$'s}\}$. Indeed, $\N$ moves to the state $q_{acc}$ when it guesses that there are no more $a$'s in the suffix to be read.
		
		We show that $\N$ is GFM. 
		Consider an MDP $\M$, and let $\mu: S\times M \to dist(A) \times M$ be a finite memory strategy such that $\mathrm{PSem}^\N_\M(\mu) = \mathrm{PSem}^\N_{\M}$.  Consider the MC $\C=\C(\M,\mu)$, and  note that since $\mu$ uses $M$ as a memory structure, we may think of $\C$ as a finite MC with positions $S\times M$, and a stochastic transition function $\Pr_\C:(S\times M)^2\rightarrow[0,1]$ that is defined $\Pr_\C(\zug{s,m},\zug{s',m'})=\sum_{c\in A_s}d(c) \cdot \Pr_{\M}(s,c,s')$ when $\mu(s,m)=\zug{d,m'}$, and otherwise $\Pr_\C(\zug{s,m},\zug{s',m'})=0$. 
		
		To conclude the GFMness of $\N$, we show that there is a strategy $\mu'\in \F(\M\times \N)$ such that $\mathrm{PSem}^\N_\M(\mu) = \mathrm{PSyn}^\N_\M(\mu')$. 	
		Recall that a strategy for $\M\times \N$ needs to take actions in $\M$ and in $\N$ simultaneously.  In the first it needs to choose an action $c\in A_s$ where $s\in S$ is the current state in $\M$, while in the latter it should choose a $\tau(s)$-successor of the current state $q$ in $\N$. The strategy $\mu'$ resolves the first according to $\mu$, and hence ignores the $\N$ component. Then, for the $\N$ component, $\mu'$ generates a run $r_\N$ of $\N$ on $w_\C$ as follows. The run $r_\N$ waits at $q_0$ (possibly forever), and moves to $q_{acc}$ only when an ergodic SCC that does not contain an $a$-labeled state is reached in $\C$; in particular, if the ergodic SCC that is reached in $\C$ contains an $a$-labeled state, then $r_\N$ does not leave $q_0$. In other words, the strategy $\mu'$ moves from $q_0$ to $q_{acc}$ only when it is certain that there are no more $a$'s in the suffix to be read.

		As the MC $\C$ is finite, then with probability 1 an ergodic component of $\C$ is reached and all of its states are traversed infinitely often. Hence, with probability 1, we have one of the following outcomes. Either the ergodic component that is reached does not contain an $a$-labeled state, or it contains an $a$-labeled state. In the former case,  the run $r_\N$ is accepting. Indeed, according to $\mu'$, the run $r_\N$ moves to $q_{acc}$ once the ergodic component that has no $a$-labeled states is reached, and thus $r_\N$ eventually gets stuck at $q_{acc}$ while reading a suffix of $w_\C$ that has no $a$'s. 
		In the latter case, $r_\N$ is stuck in $q_0$ and is thus rejecting, and with probability 1 $w_\C$ has infinitely many $a$'s. Hence, in the latter case there is a zero probability that $w_\C \in L(\N)$ and $r_\N$ is rejecting, and we are done.
		\begin{figure}[ht]
			\centering
			\includegraphics[height=2cm]{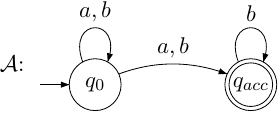}
			\caption{A GFM-NWW that is not SD. Missing transitions lead to a rejecting sink.} 
			\label{GfmNotSD}
		\end{figure}
\end{proof}

\bibliographystyle{alphaurl} 
\bibliography{ok} 

\appendix

\section{Determinization of SD-NBWs}\label{app css}

In this section we prove formally that the DBW $\A'$ is equivalent to the SD-NBW $\A$, and that if $\A$ is weak, then so is $\A'$. The following two propositions follow immediately from the definitions: 

\begin{prop}\label{pruned-css}
	Consider states $q\in Q$ and $S\in Q'$, a letter $\sigma\in \Sigma$, and transitions $\zug{q, \sigma, q'}$ and $\zug{S, \sigma, S'}$ of $\A$ and $\A'$, respectively. If $q$ is $\A$-equivalent to the states in $S$, then $q'$ is $\A$-equivalent to the states in $S'$.
\end{prop}

\begin{prop}\label{safe-succ}
	Consider a state $S$ of $\A'$ and a letter $\sigma\in \Sigma$. If $\zug{S, \sigma, S'} \in \Delta'$ and $S'\notin \alpha'$, then all the $\sigma$-successors of a state $s\in S$ are in $S'\setminus \alpha$.
\end{prop}

We can now prove the correctness of the construction:

\begin{prop}\label{A and A_D are equiv}
	The automata $\A$ and $\A'$ are equivalent.
\end{prop}

\begin{proof}
We first prove that $L(\A') \subseteq L(\A)$. Let $r_{\A'}=S_0, S_1, S_2, \ldots $ be an accepting run of $\A'$ on a  word $w = \sigma_1 \cdot \sigma_2 \cdots$. We construct an accepting run of $\A$ on $w$. Since $r_{\A'}$ is accepting, there are infinitely many positions $j_1,j_2,\ldots$ with $S_{j_i} \in \alpha'$. 

We also define $j_0 = 0$.

\noindent 
	Consider the DAG $G= \zug{V, E}$, where 
	\begin{itemize}
		\item 	 $V\subseteq Q\times \mathbb{N}$ is the union $\bigcup_{i\geq 0}(S_{j_i} \times \{i\})$.
		
		\item $E\subseteq \bigcup_{i\geq 0} (S_{j_i} \times \{i\}) \times (S_{j_{i+1}} \times \{i+1\})$ is such that for all $i\geq 0$, it holds that $E(\zug{s', i}, \zug{s, i+1})$ iff there is a finite run from $s'$ to $s$ over $w[j_i + 1, j_{i+1} ] $. Then, we label this edge by the run from $s'$ to $s$. 
	\end{itemize}
	By the definition of $\A'$, for every $j \geq 0$ and state $s_{j+1}\in S_{j+1}$, there is a state $s_j\in S_j$ such that $\zug{s_j, \sigma_j, s_{j+1}} \in \Delta$. Thus, it follows by induction that for every $i \geq 0$ and state $s_{i+1}  \in S_{j_{i+1}}$, there is a state $s_i \in S_{j_i}$ such that there is a finite run from $s_i$ to $s_{i+1}$ on $w[j_i + 1, j_{i+1} ] $. Thus, the DAG $G$ has infinitely many reachable vertices from the vertex $\zug{q_0, 0}$. Also, as the nondeterminism degree of $\A$ is finite, so is the branching degree of $G$. Thus, by K\"onig's Lemma, $G$ includes an infinite path, and the labels along the edges of this path define a run of $\A$ on $w$. Since for all $i \geq 1$, the state $S_{j_i}$  is in $\alpha'$, and so all the states in $S_{j_i}$ are in $\alpha$, this run is accepting, and we are done.

	For the other direction, assume that $w = \sigma_1 \cdot \sigma_2 \cdots \in L(\A)$, and let $r=r_0, r_1, \ldots $ be an accepting run of $\A$ on $w$.
	
	Let $S_0, S_1, S_2 \ldots $ be the run of $\A'$ on $w$, and assume, by way of contradiction, that there is a position $j\geq 0$ such that $S_j, S_{j+1}, \ldots $ is an $\alpha$-free run on the suffix $w[j+1, \infty]$. Then, an iterative application of Proposition~\ref{safe-succ} implies that all the runs of a state $s_j \in S_j$ on  $w[j+1, \infty]$ are $\alpha$-free in $\A$. Also, an iterative application of Proposition~\ref{pruned-css} implies that $r_j \sim_\A s_j$, and since $r$ is an accepting run of $\A$, it holds that $\A^{s_j}$ has an accepting run on $w[j+1, \infty]$, and we have reached a contradiction.
	
\end{proof}

It is left to prove that weakness of $\A$ is preserved in $\A'$. 

\begin{prop}\label{NWW to DWW}
	If $\A$ is an NWW, then $\A'$ is a DWW.
\end{prop}

\begin{proof}
	Assume by way of contradiction that there are reachable states $S\in \alpha'$ and $S'\notin \alpha'$, and an infinite run $r_{\A'} = S_0, S_1, S_2, \ldots$ that visits both $S$ and $S'$ infinitely often. 
	Recall that a reachable state in $Q'$ contains only $\alpha$-states of $\A$ or only $\bar{\alpha}$-states of $\A$. Hence, $S'$ contains only $\bar{\alpha}$-states of $\A$. 
	
	As in the proof of Proposition~\ref{A and A_D are equiv}, the run $r_{\A'}$ induces an infinite run $r_\A = s_0, s_1, s_2, \ldots $, where for all positions $j \geq 0$, it holds that $s_j\in S_j$. Since the run $r_{\A'}$ visits $S$ infinitely often, then $r_\A$ visits infinitely many $\alpha$-states. Likewise, since $r_{\A'}$ visits  $S'$ infinitely often, then $r_\A$ also visits infinitely many $\bar{\alpha}$-states. This contradicts the weakness of $\A$, and we are done.
\end{proof}

\section{Details of the Reduction in Theorem~\ref{sd pspace}}
\label{app tm}

We describe the technical details of the construction of $\A$.
Let $T=\zug{\Gamma,Q,\rightarrow,q_0,q_{acc},q_{rej}}$, where $\Gamma$ is the working
alphabet, $Q$ is the set of states, $\rightarrow \subseteq Q \times
\Gamma \times Q \times \Gamma \times \{L,R\}$ is the transition
relation (we use $(q,a) \rightarrow (q',b,\Delta)$ to indicate that
when $T$ is in state $q$ and it reads the input $a$ in the current
tape cell, it moves to state $q'$, writes $b$ in the current tape
cell, and its reading head moves one cell to the left/right, according
to $\Delta$), $q_0$ is the initial state, $q_{acc}$ is the accepting state, and $q_{rej}$ is the rejecting one. 

The transitions function $\rightarrow$ is defined also for the final states $q_{acc}$ and $q_{rej}$: when a computation of $T$ reaches them, it erases the tape, goes to the leftmost cell in the tape, 
and moves to the initial state $q_0$. Recall that $s: \mathbb{N} \rightarrow \mathbb{N}$ is the polynomial space function of $T$. Thus, when $T$ runs on the empty tape, it uses at most $n_0=s(0)$ cells. 

We use strings of the form $\# \gamma_1 \gamma_2 \ldots (q,\gamma_i) \ldots \gamma_{n_0}$ to encode a configuration of $T$ on a word of length at most $n_0$.
That is, a configuration starts with $\#$, and all its other letters
are in $\Gamma$, except for one letter in $Q \times \Gamma$.
The meaning of such a configuration is that the $j$'th cell in $T$, for
$1 \leq j \leq n_0$, is labeled $\gamma_j$, the reading head points
at cell $i$, and $T$ is in state $q$. For example, the initial
configuration of $T$ is
$\# (q_0,b) b \ldots b$ (with $n_0-1$ occurrences of $b$'s) where $b$
stands for an empty cell.
We can now encode a computation of $T$ by a sequence of configurations.

Let $\Sigma=\{\#\} \cup \Gamma \cup (Q \times \Gamma)$ and let
$\# \sigma_1 \ldots \sigma_{n_0} \# \sigma'_1 \ldots \sigma'_{n_0}$
be two successive configurations of $T$. We also set
$\sigma_0$, $\sigma'_0$, and $\sigma_{n_0+1}$ to $\#$.
For each triple $\zug{\sigma_{i-1},\sigma_i,\sigma_{i+1}}$ with $1
\leq i \leq n_0$, we know, by the transition relation of $T$, what
$\sigma'_i$ should be. In addition, the letter $\#$ should repeat
exactly every $n_0+1$ letters.
Let $next(\sigma_{i-1},\sigma_i,\sigma_{i+1})$ denote our
expectation for $\sigma'_i$. That is,
\begin{itemize}
\item $next(\sigma_{i-1},\sigma_i,\sigma_{i+1})=\sigma_i$ if $\sigma_i=\#$, or if non of $\sigma_{i-1},\sigma_i$ and $\sigma_{i+1}$ are in $Q\times\Gamma$.
\item $next(\sigma_{i-1},\sigma_i,\sigma_{i+1})=\#$ if $|\set{\sigma_{i-1},\sigma_i,\sigma_{i+1}}\cap Q\times \Gamma|\geq 2$, or if $|\set{\sigma_{i-1},\sigma_i,\sigma_{i+1}}\cap \set{\#}|\geq 2$~\footnote{These case describes an illegal encoding of a configuration, and there is no correct expectation of the letter that comes in the same position in the next configuration. Thus, the choice $next(\sigma_{i-1},\sigma_i,\sigma_{i+1})=\#$ is arbitrary.}.
\item
$next((q,\gamma_{i-1}),\gamma_i,\gamma_{i+1})=
next((q,\gamma_{i-1}),\gamma_i,\#)=$
\begin{center}
$\left\{
\begin{array}{ll}
\gamma_i & \mbox{If $(q,\gamma_{i-1}) \rightarrow
(q',\gamma'_{i-1},L)$}\\
(q',\gamma_i) & \mbox{If $(q,\gamma_{i-1}) \rightarrow
(q',\gamma'_{i-1},R)$}
\end{array}
\right.$
\end{center}
\item
$next(\gamma_{i-1},\gamma_i,(q,\gamma_{i+1}))=
next(\#,\gamma_i,(q,\gamma_{i+1}))=$
\begin{center}
	$\left\{
	\begin{array}{ll}
		\gamma_i & \mbox{If $(q,\gamma_{i+1}) \rightarrow
			(q',\gamma'_{i+1},R)$}\\
		(q',\gamma_i) & \mbox{If $(q,\gamma_{i+1}) \rightarrow
			(q',\gamma'_{i},L)$}
	\end{array}
	\right.$
\end{center}
\item 
$next(\#,(q,\gamma_{i}),\gamma_{i+1})=(q',\gamma')$ where $(q,\gamma_{i}) \rightarrow (q',\gamma'_{i},L)$
\footnote{
	We assume that the reading head of $T$ stays in place when it is above the leftmost cell and the transition function directs it to move left.}.
\item
$next(\gamma_{i-1},(q,\gamma_{i}),\gamma_{i+1})=
next(\gamma_{i-1},(q,\gamma_{i}),\#)=
\gamma'_i$ where $(q,\gamma_{i}) \rightarrow (q',\gamma'_{i},\Delta)$.
\end{itemize}

\noindent
Consistency with $next$ now gives us a necessary condition for a word
to encode a legal computation that uses $n_0$ tape cells.

In order to accept words that satisfy C1, namely detect a violation of the encoding, we distinguish between two types of violations: a violation of a single encoded configuration, and a violation of $next$. 
In order to detect a violation of $next$, the NWW $\A$ use its
nondeterminism and guesses a triple $\zug{\sigma_{i-1},\sigma_i,\sigma_{i+1}} \in \Sigma^3$ and guesses a position in the word, where it checks whether the three letters
to be read starting this position are $\sigma_{i-1},\sigma_i$, and
$\sigma_{i+1}$, and checks whether
$next(\sigma_{i-1},\sigma_i,\sigma_{i+1})$ is not the
letter to come $n_0+1$ letters later.
Once $\A$ sees such a violation, it goes to an accepting sink. If $next$ is respected, 
or if the guessed triple and position is not successful, then $\A$ returns to its initial state. Also, at any point that $\A$ still waits to guess a position of a triple, it can guess to return back to the initial state.

In order to detect a violation of a single encoded configuration, the NWW $\A$ accepts words that include in them a subword $x$ of length $n_0+1$, which is either $\#$-free, or is of the form $\#\cdot y$, for a finite word $y$ that does not encode a legal configuration that uses $n_0$ cells. Specifically, $y$ is not of the form $\Gamma^i \cdot (Q\times \Gamma) \cdot \Gamma^{n_0-(i+1)}$, for some $0 \leq i \leq n_0-1$.
Indeed, such words either include a long subword with no $\#$, 
or have some block of size $n_0$ that follows a $\#$, that does not encode a legal configuration of $T$ when it uses $n_0$ cells. 
If $\A$ succeeds to detect such a violation, it goes to an accepting sink, otherwise, $\A$ returns to its initial state.

In order to accept words that satisfy C2, namely detect an encoding of an accepting configuration of $T$, the NWW $\A$ guesses a position where it checks if the next $n_0+1$ letters are of the form $\# \cdot x\cdot (q_{acc},\gamma)\cdot y$, for $x,y\in \Gamma^*$ and $\gamma\in \Gamma$. 
If the guess succeeded, then $\A$ goes to the accepting sink, and otherwise it goes back to the initial state. At any point that $\A$ waits to guess a position where the accepting configuration begins, it can guess to return back to the initial state.

\section{Correctness and full details of the reduction in Theorem~\ref{NPC Buchi DBP}}
\label{app np HDnbw}

We first prove that the HD strategy $g$ defined in Proposition~\ref{A phi HD} satisfies two essential properties. Then, in Lemma~\ref{criterion for Lnm winning strategy}, we show that these properties imply that $g$ is a winning HD strategy for $\A_\varphi$.

\begin{lem}\label{Lnm strategy from p takes us to q_0}
	For all $u\in (X\cdot C)^*$ and $v\in R_{n,m}$, if $g(u)=p$, then there is a prefix $y\in (X\cdot C)^*$ of $v$ such that $g(u\cdot y)=q_0$.
\end{lem}

\begin{proof}
	Let $j\in [m]$ and $2\leq k\leq n$ be such that $v$ ends with the word $x_k\cdot c_j\cdot x_{k+1}\cdot c_j\cdots x_n\cdot c_j$, and $k$ is the minimal index that is greater than $1$ for which $x_k$ participates in $c_j$. 
	Since we assume that each of the clauses of $\varphi$ depends on at least two variables, such $k>1$ exists. 
	Let $i\in {\{0,1\}}$ be minimal with $c_j \in C_k^i$, and let $z\in (X\cdot C)^*$ be a prefix of $v$ such that $v=z\cdot x_k\cdot c_j\cdots x_n\cdot c_j$.
	If there is a prefix $y\in (X\cdot C)^*$ of $z$ such that $g(u\cdot y)=q_0$, then we are done. Otherwise, $g(u\cdot z)=p$. By definition of $g$ and the choice of $k$, we know that $g(u\cdot z\cdot x_k)=q^i_k$, where the assignment $x_k=i$ satisfies $c_j$. Thus, if we take $y=z\cdot x_k\cdot c_j$, then $g(u\cdot y)=q_0$, and $y$ is a prefix of $v$.
\end{proof}

\begin{lem}\label{Lnm strategy from q_0 takes us to p}
	For all $u\in (X\cdot C)^*$ and $v\in R_{n,m}$, if $g(u)=q_0$, there is a prefix $z\in (X\cdot C)^*$ of $v$ such that $g(u\cdot z)=p$.
\end{lem}

\begin{proof}
	This follows immediately from the fact that $D_{n,m}$ is a DFW that recognizes $R_{n,m}$ and $p$ is the only accepting state of $D_{n,m}$. Thus, we may take $z$ to be the minimal prefix of $v$ that is in $R_{n,m}$.
\end{proof}

Recall that an HD strategy $g:\Sigma^*\rightarrow Q$ has to agree with the transitions of $\A_\varphi$. That is, 
for all $w\in (X\cdot C)^*$, $x_k\in X$, and $c_j\in C$, it holds that $(g(w),x_k,g(w\cdot x_k))$ and $(g(w\cdot x_k),c_j,g(x\cdot x_k\cdot c_j))$ are in $\Delta_\varphi$. In addition, if $g$ satisfies the conditions in Lemmas \ref{Lnm strategy from p takes us to q_0} and \ref{Lnm strategy from q_0 takes us to p}, we say that $g$ {\em supports a $(p,q_0)$-circle}.  

\begin{lem}\label{criterion for Lnm winning strategy}
	If $g:\Sigma^*\rightarrow Q$ is consistent with $\Delta_\varphi$ and supports a $(p,q_0)$-circle, then for all words $w\in L_{n,m}$, the run $g(w)$ is accepting.
\end{lem}

\begin{proof}
	Consider a word $w\in L_{n,m}=(R_{n,m})^\omega$. Observe that if $w'\in (X\cdot C)^\omega$ is a suffix of $w$, then $w'\in L_{n,m}$, and hence has a prefix in $R_{n,m}$. Thus, if $g$ supports a $(p,q_0)$-circle, there exist $y_1,z_1\in (X\cdot C)^*$, such that $y_1\cdot z_1$ is a prefix of $w$, $g(y_1)=q_0$, and $g(y_1\cdot z_1)=p$. Let $w' \in (X\cdot C)^\omega$ be the suffix of $w$ with $w=y_1\cdot z_1\cdot w'$. By the above, $w'\in L_{n,m}$, and we can now apply again the assumption on $g$ to obtain $y_2,z_2\in (X\cdot C)^*$ such that $y_2\cdot z_2$ is a prefix of $w'$, $g(y_1\cdot z_1\cdot y_2)=q_0$, and $g(y_1\cdot z_1\cdot y_2\cdot z_2)=p$. By iteratively applying this argument, we construct $\{y_i,z_i:i\geq 1\}\subseteq (X\cdot C)^*$, such that $w^i=y_1\cdot z_1\cdot y_2\cdot z_2\cdots y_{i-1}\cdot z_{i-1}\cdot y_i$ is a prefix of $w$, and $g(w^i)=q_0$, for all $i\geq 1$. We conclude that $q_0\in \inf (g(w))$, and hence $g(w)$ is accepting.
\end{proof}

It is easy to see that there is a correspondence between assignments to the variables in $X$ and deterministic prunings of $ \A_\varphi $. Indeed, a pruning of $p$ amounts  to choosing, for each $ k\in [n]$, a value $ i_k\in {\{0,1\}} $: the assignment $x_k=i_k $ corresponds to keeping the transition $ \zug{p,x_k,q^{i_k}_k} $ and removing the transition $ \zug{p,x_k,q^{\lnot i_k}_k} $. For an assignment $a:X \rightarrow {\{0,1\}}$, we denote by $\A_\varphi^{a}$ the deterministic pruning of $\A_\varphi$ that is associated with $a$. We prove that $a$ satisfies $\varphi$ iff $\A_\varphi^{a}$ is equivalent to $\A_\varphi$. Thus, the number of deterministic prunings of $\A_\varphi$ that result in a DBW equivalent to $\A_\varphi$, equals to the number of assignments that satisfy $\varphi$. In particular, $\varphi$ is satisfiable iff $\A_\varphi$ is DBP. 

\begin{prop}
For every assignment $a	:X \rightarrow {\{0,1\}}$, we have that $ L(\A^{a}_\varphi)=L(\A_\varphi) $ iff $ \varphi  $ is satisfied by $a$.
\end{prop}

\begin{proof}
	Assume first that $\varphi$ is not satisfied by $a$. We prove that $L_{n,m}\neq L(\A^{a}_\varphi)$. Let $j\in [m]$ be such that $c_j$ is not satisfied by $a$. I.e, for all $k\in [n]$ the assignment $x_k=i_k$ does not satisfy $c_j$. Since $q^i_k$ is reachable in $\A_\varphi^{a}$ iff $i=i_k$, and all $c_j$-labeled transitions from $\{q^{i_k}_k : k\in[n]\}$ are to $p$, it follows that the run of $\A_\varphi^{a}$ on $\{x_1\cdot c_j\cdot x_2\cdot c_j\cdots x_n\cdot c_j\}^\omega$ never visits $q_0$, and hence is rejecting. Thus, $(x_1\cdot c_j\cdot x_2\cdot c_j\cdots x_n\cdot c_j)^\omega\in L_{n,m}\setminus L(\A_\varphi^{a})$.
	
	For the other direction, we assume that $a$ satisfies $\varphi$ and prove that $L(\A_\varphi^{a})=L_{n,m}$. Let $g^{a}:\Sigma^*\rightarrow Q$ be the memoryless strategy that correspond to the pruning $\A^{a}_\varphi$. By Lemma~\ref{criterion for Lnm winning strategy}, it is sufficient proving that $g^a$ supports a $(p,q_0)$-circle. Note that every strategy for $L_{n,m}$ satisfies Lemma~\ref{Lnm strategy from q_0 takes us to p}. Indeed, the proof only uses the fact that $D_{n,m}$ is a DFW that recognizes $R_{n,m}$ with a single accepting state $p$. Thus, we only need to prove that $g^a$ satisfies Lemma~\ref{Lnm strategy from p takes us to q_0}. That is, for all $u\in (X\cdot C)^*$ and $v\in R_{n,m}$, if $g^{a}(u)=p$, then there is a prefix $y\in (X\cdot C)^*$ of $v$, such that $g^{a}(u\cdot y)=q_0$. Consider such words $u$ and $v$, and let $j\in [m]$ be such that $c_j$ is the last letter of $v$. 
	
		Let $ k\in [n]$ be the minimal index for which $c_j\in C^{i_k}_k$, 
	and let $z\in (X\cdot C)^*$ be a prefix of $v$ such that $v=z\cdot x_k\cdot c_j\cdot x_{k+1}\cdot c_j\cdots x_n\cdot c_j$. If there exists a prefix $y\in (X\cdot C)^*$ of $ z$ such that $g^{a}(u\cdot y)=q_0$, then we are done. Otherwise, the finite run of $\A_\varphi^{a}$ on $ z $ from $ p $, returns back to $ p $, and hence $g^{a}(u\cdot z)=p$. Now $g^{a}(u\cdot z\cdot x_k)=q^{i_k}_k$, and since $x_k=i_k$ satisfies $c_j$ we have $g^{a}(u\cdot z\cdot x_k\cdot c_j)=q_0$. Thus, we may take $y=z\cdot x_k\cdot c_j$ which is a prefix of $v$, and we are done.
\end{proof}

\end{document}